\RequirePackage{fix-cm}
\documentclass{svjour3}                     
\smartqed  
\usepackage{graphicx}
\usepackage{pslatex}
\usepackage{amsfonts,color}
\usepackage{amssymb,amsmath,latexsym}

\makeatletter
\def\inmod#1{\allowbreak\mkern5mu({\operator@font mod}\,\,#1)}
\makeatother

\newcommand{\wt}{{\mathrm{wt}}}

\newcommand{\tr}{{\mathrm{Tr}}}
\newcommand{\gf}{{\mathrm{GF}}}

\newcommand{\C}{{\mathcal{C}}}

\newcommand{\bc}{{\mathbf{c}}}

\usepackage[justification=centering]{caption}

\begin{document}

\title{A Class of Linear Codes with a Few Weights 
}


\author{Can Xiang$^1$ \and  Chunming Tang$^2$ \and Keqin Feng$^3$ \thanks{The research of K. Feng was supported by NSFC No. 11471178, 11571007 and the Tsinghua National Lab. for Information Science and Technology.}}


\institute{C. Xiang \at
                College of Mathematics and Information Science, Guangzhou University, Guangzhou 510006, China \\
              \email{cxiangcxiang@hotmail.com} \and
              C. Tang \at
                School of Mathematics and Information, China West Normal University, Sichuan Nanchong, 637002, China \\
              \email{tangchunmingmath@163.com} \and
              K. Feng \at
                Department of Mathematical Sciences, Tsinghua University, Beijing, 100084, China \\
              \email{kfeng@math.tsinghua.edu.cn}
}

\date{Received: date / Accepted: date}

\maketitle

\begin{abstract}
Linear codes have been an interesting subject of study for many years, as linear codes with few weights have applications in secrete sharing, authentication codes, association schemes, and strongly regular graphs. In this paper, a class of linear codes with a few weights over the finite field $\gf(p)$ are presented and their weight distributions are also determined, where $p$ is an odd prime.
Some of the linear codes obtained are optimal in the sense that they meet certain bounds on linear codes.

\end{abstract}

\section{Introduction}\label{sec-intro}

Throughout this paper, let $p$ be an odd prime and let $q=p^m$ for some positive integer $m$.
An $[n,\, k,\,d]$ code $\C$ over $\gf(p)$ is a $k$-dimensional subspace of $\gf(p)^n$ with minimum
(Hamming) distance $d$.  Let $A_i$ denote the number of codewords with Hamming weight $i$ in a code
$\C$ of length $n$. The {\em weight enumerator} of $\C$ is defined by
$
1+A_1z+A_2z^2+ \cdots + A_nz^n.
$
The {\em weight distribution} $(1,A_1,\ldots,A_n)$ is an important research topic in coding theory,
as it contains crucial information as to estimate the error correcting capability and the probability of
error detection and correction with respect to some algorithms.
A code $\C$ is said to be a $t$-weight code  if the number of nonzero
$A_i$ in the sequence $(A_1, A_2, \cdots, A_n)$ is equal to $t$.

Let $D=\{d_1, \,d_2, \,\ldots, \,d_n\} \subseteq \gf(q)$.
Let $\tr$ denote the trace function from $\gf(q)$ onto $\gf(p)$ throughout
this paper. We define a linear code of
length $n$ over $\gf(p)$ by
\begin{eqnarray}\label{eqn-maincode}
\C_{D}=\{(\tr(xd_1), \tr(xd_2), \ldots, \tr(xd_n)): x \in \gf(q)\},
\end{eqnarray}
and call $D$ the \emph{defining set} of this code $\C_{D}$. By definition, the
dimension of the code $\C_D$ is at most $m$.

This construction is generic in the sense that many classes of known codes
could be produced by properly selecting the defining set $D \subseteq \gf(q)$. If
the defining set $D$ is well chosen, some optimal linear codes with few weights can be obtained. This
construction technique was employed in
\cite{Ding15}, \cite{DingDing1},  \cite{DingDing2}, \cite{Ding09}, \cite{DLN}, \cite{DN07} and \cite{WDX2015}
for
obtaining linear codes with a few weights. For more details, we refer interested readers to \cite{Ding20152,Mesnager2015,ZLFH2015,TLQZH2015,TQH2015,QTH2015} and the references therein.

The purpose of this paper is to construct a class of linear codes over $\gf(p)$ with a few
nonzero weights using this generic construction method, and determine their weight distributions. Some of the linear
codes obtained in this paper are optimal in the sense that they meet some bounds on linear codes.
The linear codes with a few weights presented in this paper have applications also in
secret sharing \cite{ADHK,CDY05,YD06}, authentication codes \cite{CX05}, combinatorial designs and graph
theory \cite{CG84,CK85}, and association schemes \cite{CG84}, in addition to their applications in consumer electronics, communication and data storage systems.

The remainder of this paper is organized as follows. Section \ref{sec-pr} introduces some basic notations and results of group characters, Gauss
sums, exponential sums and cyclotomic fields which will be needed in subsequent sections. Section \ref{sec-main} presents a class of linear codes with a few weights and the proofs of their parameters are given in Section \ref{sec-proof}. Section \ref{sec-concluding} summarizes this paper.

\section {Preliminaries}\label{sec-pr}
In this section, we state some notations and basic facts on group characters, Gauss
sums, exponential sums and cyclotomic fields. These
results will be used later in this paper.

\subsection{Some notations fixed throughout this paper}

For convenience, we adopt the following notations unless otherwise stated in this paper.
\begin{description}
  \item[$\bullet$]  $p^*=(-1)^{(p-1)/2}p$.
  \item[$\bullet$]  $\zeta_p=e^{\frac{2\pi \sqrt{-1}}{p}}$, a primitive $p$-th
root of unity.
  \item[$\bullet$]  $L(x)=x^{p^2}+x.$
  \item[$\bullet$]  $\textup{Im}(L)=\{L(x):x\in \gf(q)\}.$
  \item[$\bullet$]  $x_b\in \gf(q)$ denotes a solution of the equation $L(x)=-b^p$
if $b\in \textup{Im}(L)$.
  \item[$\bullet$]  $m_p=m \bmod{p} \in \{0,1,...,p-1\}$,
  the least non-negative residue modulo $m$.
  \item[$\bullet$]  $\textup{SQ}$ and $\textup{N\textup{SQ}}$ denote the set of all  squares and nonsquares in $\gf(p)^{*}$, respectively.
  \item[$\bullet$]  $\eta$ and $\bar{\eta}$ are the quadratic characters of $\gf(q)^{*}$ and  $\gf(p)^{*}$, repsectively. We extend these quadratic characters
by letting $\eta(0)=0$ and $\bar{\eta}(0)=0$.
\end{description}

\subsection{Group characters and Gauss sums}

An {\em additive character} of $\gf(q)$ is a nonzero function $\chi$
from $\gf(q)$ to the set of nonzero complex numbers such that
$\chi(x+y)=\chi(x) \chi(y)$ for any pair $(x, y) \in \gf(q)^2$.
For each $b\in \gf(q)$, the function
\begin{eqnarray}\label{dfn-add}
\chi_b(c)=\zeta_p^{\tr(bc)} \ \ \mbox{ for all }
c\in\gf(q)
\end{eqnarray}
defines an additive character of $\gf(q)$. When $b=0$,
$\chi_0(c)=1 \mbox{ for all } c\in\gf(q),
$
and is called the {\em trivial additive character} of
$\gf(q)$. The character $\chi_1$ in (\ref{dfn-add}) is called the
{\em canonical additive character} of $\gf(q)$.
It is known that every additive character of $\gf(q)$ can be
written as $\chi_b(x)=\chi_1(bx)$ \cite[Theorem 5.7]{LN}.

The Gauss sum $G(\eta, \chi_1)$ over $\gf(q)$ is defined by
\begin{eqnarray}
G(\eta, \chi_1)=\sum_{c \in \gf(q)^*} \eta(c) \chi_1(c) = \sum_{c \in \gf(q)} \eta(c) \chi_1(c)
\end{eqnarray}
and
the Gauss sum $G(\bar{\eta}, \bar{\chi}_1)$ over $\gf(p)$ is defined by
\begin{eqnarray}
G(\bar{\eta}, \bar{\chi}_1)=\sum_{c \in \gf(p)^*} \bar{\eta}(c) \bar{\chi}_1(c)
= \sum_{c \in \gf(p)} \bar{\eta}(c) \bar{\chi}_1(c),
\end{eqnarray}
where $\bar{\chi}_1$ is the canonical additive characters of $\gf(p)$.

The following three lemmas are proved in \cite[Theorem 5.15 and Theorem 5.33]{LN} and \cite[lemma 7]{DingDing2}, respectively.

\begin{lemma}\label{lem-32A1}
With the symbols and notations above, we have
$$
G(\eta, \chi_1)=(-1)^{m-1} \sqrt{-1}^{(\frac{p-1}{2})^2 m} \sqrt{q}
$$
and
$$
G(\bar{\eta}, \bar{\chi}_1)= \sqrt{-1}^{(\frac{p-1}{2})^2 } \sqrt{p}=\sqrt{p*}.
$$
\end{lemma}

\begin{lemma}\label{lem-32A2}
Let $\chi$ be a nontrivial additive character of $\gf(q)$ with $q$ odd, and let
$f(x)=a_2x^2+a_1x+a_0 \in \gf(q)[x]$ with $a_2 \ne 0$. Then
$$
\sum_{c \in \gf(q)} \chi(f(c)) = \chi(a_0-a_1^2(4a_2)^{-1}) \eta(a_2) G(\eta, \chi).
$$
\end{lemma}

\begin{lemma}\label{lem-bothcharac}
If $m \ge 2$ is even, then $\eta(y)=1$ for each $y \in \gf(p)^*$.
If $m \ge 1$ is odd, then  $\eta(y)=\bar{\eta}(y)$ for each $y \in \gf(p)$.
\end{lemma}

\subsection{A type of exponential sums}

For any $a$ and $b$ in $\gf(q)$, we define the following exponential sum
\begin{eqnarray*}\label{eqn-esum}
S(a,b)=\sum_{x \in \gf(q)} \chi_1\left(a x^{p+1}+bx\right)
\end{eqnarray*}
in this paper. To prove our main results, we need the values of the sum $S(a,b)$ and the help of a number of lemmas that are proved in \cite[Theorem 1, Theorem 2, Theorem 2.3]{Coulter}.

\begin{lemma}\label{lem-Coulter1}

Let $m$ be odd, $f(x)=a^{p}x^{p^2}+ax\in F_{q}[x]$ and $b\in \gf(q)$. Then $f(x)$ is a permutation polynomial over $\gf(q)$ and
\begin{eqnarray*}
S(a,b) =
\sqrt{p^*}^m \eta(a)\chi_1(-ax_{a,b}^{p+1}),
\end{eqnarray*}
where $p*=(-1)^{\frac{p-1}{2}}p$ and $x_{a,b}$ is the unique solution of the equation $f(x)=-b^p$.

Particularly, $S(a,0) = \sqrt{p^*}^m \eta(a)$.
\end{lemma}

\begin{lemma}\label{lem-Coulter2}

Let $m$ be even, $f(x)=a^{p}x^{p^2}+ax\in F_{q}[x]$ and $b\in \gf(q)$. There are two cases.
\begin{enumerate}
  \item If $a^{\frac{q-1}{p+1}}\neq (-1)^{\frac{m}{2}}$, then $f(x)$ is a permutation polynomial over $\gf(q)$. Let $x_{a,b}$ be the unique solution of the equation $f(x)=-b^p$. Then
      $$
      S(a,b)=(-1)^{\frac{m}{2}}p^{m/2}\chi_1(-ax_{a,b}^{p+1})
      $$
  \item If $a^{\frac{q-1}{p+1}}=(-1)^{\frac{m}{2}}$, then $f(x)$ is not a permutation polynomial over $\gf(q)$. We have $S(a,b)=0$ unless the the equation $f(x)=-b^p$ is solvable. If this equation is solvable, with solution $x_{a,b}$ say, then

      $$
      S(a,b)=-(-1)^{\frac{m}{2}}p^{m/2+1}\chi_1(-ax_{a,b}^{p+1}).
      $$
\end{enumerate}

Particularly,
\begin{eqnarray*}
S(a,0) =
\left\{ \begin{array}{ll}
(-1)^{\frac{m}{2}}p^{\frac{m}{2}}& \mbox{ if $a^{\frac{q-1}{p+1}}\neq (-1)^{\frac{m}{2}}$,} \\
(-1)^{\frac{m}{2}+1}p^{\frac{m}{2}+1}& \mbox{ if $a^{\frac{q-1}{p+1}} = (-1)^{\frac{m}{2}}$.}
\end{array}
\right.
\end{eqnarray*}
\end{lemma}

\subsection{Cyclotomic fields}
In this subsection, we state some basic facts on cyclotomic fields. These results will be used in the rest of this paper.

Let $\mathbb{Z}$ be the rational integer ring and $Q$ be the rational field.
Some results on
cyclotomic field $Q(\zeta_p)$ \cite{IR1990} are given in the following lemma.
\begin{lemma}\label{cyclo} We have the following basic facts.
\begin{enumerate}
 \item
The ring of integers
in $K=Q(\zeta_p)$ is $\mathcal{O}_K=
\mathbb{Z}[\zeta_p]$ and $\{
\zeta_p^{~i}: 1\leq i\leq p-1\}$
is an integral basis of $\mathcal{O}_K$.
\item
The field extension $K/Q$
is Galois of degree $p-1$ and the Galois
group $Gal(K/Q)=\{\sigma_a:
a\in (\mathbb{Z}/p\mathbb{Z})^{*}\}$, where
the automorphism $\sigma_a$ of $K$ is defined by
$\sigma_a(\zeta_p)=\zeta_p^a$.
\item
The field $K$ has a unique
quadratic subfield $L=Q(\sqrt{p^*})$. For $1\leq a\leq p-1$,
$\sigma_a(\sqrt{p^*}) =
\bar{\eta}(a)\sqrt{p^*}$. Therefore, the Galois group
$Gal(L/Q)$ is $\{1,\sigma_{\gamma}\}$, where
$\gamma$ is any quadratic nonresidue in
$\gf(p)$.
\end{enumerate}
\end{lemma}

\section{The linear codes with a few weights}\label{sec-main}

We only describe the codes and introduce their parameters in this section. The proofs of their parameters will be
given in Section \ref{sec-proof}.

In this paper, the defining set $D$ of the code $\C_D$ of (\ref{eqn-maincode})  is given by
\begin{eqnarray}\label{eqn-defsetD}
D=\{x \in \gf(q)^*: \tr(x^{p+1}-x)=0\}.
\end{eqnarray}

When $p=2$, the weight distribution of the code $\C_D$ was settled in \cite{XC}. In this paper, we study the code
$\C_D$ for $p$ being an odd prime.

The following three theorems are the main results of this paper.
\begin{theorem}\label{thm-five}
Let $m\geq 3$ be odd, and let $D$ be defined in (\ref{eqn-defsetD}). Then the set $\C_D$ of (\ref{eqn-maincode}) is an $[n,m]$ linear code over $\gf(p)$ with at most five weights and the weight distribution in Tables \ref{tab-five0} and \ref{tab-five1},  where
\begin{eqnarray}\label{eqn-codelenth1}
n =
p^{m-1}-1+(-1)^{\frac{p-1}{2}\frac{m-1}{2}}\bar{\eta}(m_{p})p^{\frac{m-1}{2}}.
\end{eqnarray}

\end{theorem}

\begin{table}[ht]
\begin{center}
\caption{The weight distribution of $\C_D$ of Theorem \ref{thm-five} when $m_p=0$}\label{tab-five0}
\begin{tabular}{|c|c|} \hline
Weight $w$ &  Multiplicity $A_w$  \\ \hline
$0$          &  $1$ \\ \hline
$(p-1)p^{m-2}$  & $p^{m-1}-1$\\ \hline
$(p-1)(p^{m-2}-(-1)^{\frac{p-1}{2}\frac{m-1}{2}}p^{\frac{m-3}{2}})$  & $\frac{p-1}{2}(p^{m-2}+(-1)^{\frac{p-1}{2}\frac{m-1}{2}}p^{\frac{m-1}{2}})$ \\ \hline
$(p-1)(p^{m-2}+(-1)^{\frac{p-1}{2}\frac{m-1}{2}}p^{\frac{m-3}{2}})$  & $\frac{p-1}{2}(p^{m-2}-(-1)^{\frac{p-1}{2}\frac{m-1}{2}}p^{\frac{m-1}{2}})$\\ \hline
$(p-1)p^{m-2}+(-1)^{\frac{p-1}{2}\frac{m-1}{2}}p^{\frac{m-3}{2}}$  & $\frac{1}{2}(p-1)^2p^{m-2}$ \\ \hline
$(p-1)p^{m-2}-(-1)^{\frac{p-1}{2}\frac{m-1}{2}}p^{\frac{m-3}{2}}$  & $\frac{1}{2}(p-1)^2p^{m-2}$\\ \hline
\end{tabular}
\end{center}
\end{table}

\begin{table}[ht]
\begin{center}
\caption{The weight distribution of $\C_D$  of Theorem \ref{thm-five} when $m_p\neq 0$}\label{tab-five1}
\begin{tabular}{|c|c|} \hline
Weight $w$ &  Multiplicity $A_w$  \\ \hline
$0$          &  $1$ \\ \hline
$(p-1)p^{m-2}$  & $p^{m-2}-1+\bar{\eta}(m_p)(-1)^{\frac{p-1}{2}\frac{m-1}{2}}(p-1)p^{\frac{m-3}{2}}$\\ \hline
$(p-1)p^{m-2}+\bar{\eta}(m_p)(-1)^{\frac{p-1}{2}\frac{m-1}{2}}p^{\frac{m-1}{2}}$  & $p^{m-2}(p-1)-\bar{\eta}(m_p)(-1)^{\frac{p-1}{2}\frac{m-1}{2}}(p-1)p^{\frac{m-3}{2}}$ \\ \hline
$(p-1)p^{m-2}+\bar{\eta}(m_p)(-1)^{\frac{p-1}{2}\frac{m-1}{2}}(p+1)p^{\frac{m-3}{2}}$  & $\frac{1}{2}(p-1)(p-2)p^{\frac{m-3}{2}}(p^{\frac{m-1}{2}}-\bar{\eta}(m_p)(-1)^{\frac{p-1}{2}\frac{m-1}{2}})$\\ \hline
$(p-1)p^{m-2}+\bar{\eta}(m_p)(-1)^{\frac{p-1}{2}\frac{m-1}{2}}(p-1)p^{\frac{m-3}{2}}$  & $\frac{p-1}{2}(p^{m-1}-\bar{\eta}(m_p)(-1)^{\frac{p-1}{2}\frac{m-1}{2}}p^{\frac{m-1}{2}})$ \\ \hline
$(p-1)p^{m-2}+\bar{\eta}(m_p)(-1)^{\frac{p-1}{2}\frac{m-1}{2}}p^{\frac{m-3}{2}}$  & $(p-1)p^{m-2}+\bar{\eta}(m_p)(-1)^{\frac{p-1}{2}\frac{m-1}{2}}(p-1)^2p^{\frac{m-3}{2}}$\\ \hline
\end{tabular}
\end{center}
\end{table}


\begin{example}
Let $(p,m)=(3,5)$. Then the code $\C_D$ has parameters $[71, 5, 42]$ and weight enumerator
$1+30z^{42}+60z^{45}+90z^{48}+42z^{51}+20z^{54}$, which is verified by a Magma program.
\end{example}
\begin{example}
Let $(p,m)=(3,9)$. Then the code $\C_D$ has parameters $[6560, 9, 4320]$  and weight enumerator
$1+2268z^{4320}+4374z^{4347}+6560z^{4374}+4374z^{4401}+2106z^{4428}$, which is verified by a Magma program.
\end{example}

\begin{remark}
The code $\C_D$ of Theorem \ref{thm-five} is a five-weight linear code except in the following two cases:
\begin{enumerate}
  \item When $m=3$ and $m_p=0$, the frequency of the weight $(p-1)(p^{m-2}-(-1)^{\frac{p-1}{2}\frac{m-1}{2}}p^{\frac{m-3}{2}})$ of Table \ref{tab-five0} turns out to be $0$. Thus, the code $\C_D$ has four nonzero weights.
  \item When $m=3$, $m_p\neq 0$ and $p\equiv 2 \pmod{3}$,  the frequency of the weight $(p-1)p^{m-2}$ of Table \ref{tab-five1} turns out to be $0$. Thus, the code $\C_D$ has only four nonzero weights.
\end{enumerate}
\end{remark}
\begin{example}
Let $(p,m)=(3,3)$. Then the code $\C_D$ has parameters $[8, 3, 4]$ and weight enumerator
$1+6z^{4}+6z^{5}+8z^{6}+6z^{7}$, which is verified by a Magma program. This code is almost
optimal, since the optimal linear code has parameters $[8, 3, 5]$.
\end{example}
\begin{example}
Let $(p,m)=(5,3)$. Then the code $\C_D$ has parameters $[19, 3, 14]$ and weight enumerator
$1+36z^{14}+24z^{15}+60z^{16}+4z^{19}$, which is verified by a Magma program. This code is
optimal.
\end{example}

\begin{theorem}\label{thm-three}
Let $m\geq 2$ be even and $m\equiv 2 \pmod{4}$, and let $D$ be defined in (\ref{eqn-defsetD}). Then the set $\C_D$ of (\ref{eqn-maincode}) is an $[n,m]$ linear code over $\gf(p)$ with at most three weights and the weight distribution in Table \ref{tab-three0} and Table \ref{tab-three1}, where
\begin{eqnarray}\label{eqn-codelenth2}
n =
\left\{ \begin{array}{ll}
p^{m-1}-1-(p-1)p^{\frac{1}{2}(m-1+(-1)^{\frac{m}{2}})} & \mbox{ if $m_p=0 $,} \\
p^{m-1}-1+p^{\frac{1}{2}(m-1+(-1)^{\frac{m}{2}})} & \mbox{ if $m_p \neq 0$.}
\end{array}
\right.
\end{eqnarray}
\end{theorem}

\begin{table}[ht]
\begin{center}
\caption{The weight distribution of $\C_D$  of Theorem \ref{thm-three} when $m_p=0$} \label{tab-three0}
\begin{tabular}{|c|c|} \hline
Weight $w$ &  Multiplicity $A_w$  \\ \hline
$0$          &  $1$ \\ \hline
$(p-1)p^{m-2}$  & $p^{m-2}-(p-1)p^{\frac{m}{2}-1}-1$ \\ \hline
$(p-1)(p^{m-2}-p^{\frac{m}{2}-1})$  & $(p-1)(2p^{m-2}+p^{\frac{m}{2}-1})$\\ \hline
$(p-1)p^{m-2}-(p-2)p^{\frac{m}{2}-1}$  & $(p-1)^2 p^{m-2}$\\ \hline
\end{tabular}
\end{center}
\end{table}

\begin{table}[ht]
\begin{center}
\caption{The weight distribution of $\C_D$  of Theorem \ref{thm-three} when $m_p \neq 0$}\label{tab-three1}
\begin{tabular}{|c|c|} \hline
Weight $w$ &  Multiplicity $A_w$  \\ \hline
$0$          &  $1$ \\ \hline
$(p-1)p^{m-2}$  & $p^{m-2}+\frac{p-1}{2}(p^{m-1}+p^{\frac{m}{2}})-1$ \\ \hline
$(p-1)p^{m-2}+p^{\frac{m}{2}-1}$  & $(p-1)(2p^{m-2}-p^{\frac{m}{2}-1})$\\ \hline
$(p-1)p^{m-2}+2p^{\frac{m}{2}-1}$  & $\frac{(p-1)(p-2)}{2}(p^{m-2}-p^{\frac{m}{2}-1})$ \\ \hline
\end{tabular}
\end{center}
\end{table}

\begin{example}
Let $(p,m)=(3,6)$. Then the code $\C_D$ has parameters $[224, 6, 144]$ and weight enumerator
$1+342z^{144}+324z^{153}+62z^{162}$, which is verified by a Magma program.
\end{example}

\begin{example}
Let $(p,m)=(5,6)$. Then the code $\C_D$ has parameters $[3149, 6, 2500]$ and weight enumerator
$1+7124z^{2500}+2525z^{4900}+2550z^{3600}$, which is verified by a Magma program.
\end{example}

\begin{remark}
The code $\C_D$ of Theorem \ref{thm-three} has three weights except in the case $m=2$, as the frequency
of the weight $(p-1)p^{m-2}+2p^{\frac{m}{2}-1}$ of Table \ref{tab-three1} turns out to be $0$.
Hence, The code $\C_D$ of Theorem \ref{thm-three} is a two-weight linear code if and only if $m=2$.
\end{remark}

\begin{example}
Let $(p,m)=(3,2)$. Then the code $\C_D$ has parameters $[3, 2, 2]$ and weight enumerator
$1+6z^{2}+2z^{3}$, which is verified by a Magma program. This code is optimal.
\end{example}

\begin{theorem}\label{thm-four}
Let $m\geq 6$ be even and $m\equiv 0 \pmod{4}$, and let $D$ be defined in (\ref{eqn-defsetD}). Then the set $\C_D$ of (\ref{eqn-maincode}) is a four-weight linear code over $\gf(p)$ with the parameter $[n, m]$ and the weight distribution in Tables \ref{tab-four0} and \ref{tab-four1}, where
$n$ is defined by (\ref{eqn-codelenth2}).
\end{theorem}

\begin{table}[ht]
\begin{center}
\caption{The weight distribution of $\C_D$  of Theorem \ref{thm-four} when $m_p=0$ } \label{tab-four0}
\begin{tabular}{|c|c|} \hline
Weight $w$ &  Multiplicity $A_w$  \\ \hline
$0$          &  $1$ \\ \hline
$p^{m-2}(p-1)-(p-1)^2p^{\frac{m}{2}-1}$  & $(p^2-1)p^{m-2}$ \\ \hline
$p^{m-2}(p-1)$  & $p^{m-4}-(p-1)p^{\frac{m}{2}-2}-1$\\ \hline
$(p-1)p^{\frac{m}{2}}(p^{\frac{m}{2}-2}-1)$  & $(p-1)(2p^{m-4}+p^{\frac{m}{2}-2})$ \\ \hline
$(p-1)p^{m-2}-(p-2)p^{\frac{m}{2}}$  & $(p-1)^2p^{m-4}$\\ \hline
\end{tabular}
\end{center}
\end{table}

\begin{table}[ht]
\begin{center}
\caption{The weight distribution of $\C_D$  of Theorem \ref{thm-four} when $m_p\neq 0$}\label{tab-four1}
\begin{tabular}{|c|c|} \hline
Weight $w$ &  Multiplicity $A_w$  \\ \hline
$0$          &  $1$ \\ \hline
$(p-1)(p^{\frac{m}{2}-1}+p^{m-2})$  & $p^m-p^{m-2}$ \\ \hline
$(p-1)p^{m-2}$  & $p^{m-4}+\frac{p-1}{2}(p^{\frac{m}{2}-1}+p^{m-3})-1$\\ \hline
$(p-1)p^{m-2}+p^{\frac{m}{2}}$  & $(p-1)(2p^{m-4}-p^{\frac{m}{2}-2})$\\ \hline
$(p-1)p^{m-2}+2p^{\frac{m}{2}}$  & $\frac{1}{2}(p-1)(p-2)(p^{m-4}-p^{\frac{m}{2}-2})$ \\ \hline
\end{tabular}
\end{center}
\end{table}

\begin{example}
Let $(p,m)=(3,8)$. Then the code $\C_D$ has parameters $[2267, 8, 1458]$ and weight enumerator
$1+350z^{1458}+5832z^{1512}+306z^{1539}+32z^{1620}$, which is verified by a Magma program.
\end{example}

\begin{example}
Let $(p,m)=(5,8)$. Then the code $\C_D$ has parameters $[78749, 8, 62500]$ and weight enumerator
$1+7124z^{62500}+375000z^{63000}+4900z^{63125}+3600z^{63750}$, which is verified by a Magma program.
\end{example}

\section{The proofs of the main results}\label{sec-proof}

Our task of this section is to prove Theorems \ref{thm-five}, \ref{thm-three} and \ref{thm-four}, respectively. To this end, we shall prove a few more auxiliary results before proving the main results of this paper.

\subsection{Some auxiliary results}

\begin{lemma}\label{lem-32B1}
With the symbols and notations above, we have
\begin{eqnarray*}
\sum_{c \in \gf(p)^*} S(c,-c)=\left\{ \begin{array}{ll}
(-1)^{\frac{p-1}{2}}\bar{\eta}(m_p)\sqrt{p^*}^{m+1}                               & \mbox{ if $m$ is odd,} \\
-(p-1)p^{\frac{1}{2}(m+1+(-1)^{\frac{m}{2}})}& \mbox{ if $m$ is even and $m_p=0$,} \\
p^{\frac{1}{2}(m+1+(-1)^{\frac{m}{2}})}& \mbox{ if $m$ is even and $m_p\neq 0$}.
\end{array}
\right.
\end{eqnarray*}
\end{lemma}

\begin{proof}
For any $c\in \gf(p)^*$, it is easily seen that $f(x)=c^{p}x^{p^2}+cx=-(-c)^p$ is solvable and $\frac{1}{2}$ is its solution.

By definition, Lemmas \ref{lem-Coulter1} and \ref{lem-Coulter2}, we have
\begin{eqnarray*}
& &\sum_{c \in \gf(p)^*} S(c,-c)\\
& &=\sum_{c \in \gf(p)^*} \sum_{x \in \gf(q)} \chi_1\left(c x^{p+1}-cx\right)\\
& &=
\left\{ \begin{array}{ll}
\sum_{c \in \gf(p)^*}\sqrt{p^*}^m \eta(c)\chi_1(-\frac{c}{4})                               & \mbox{ if $m$ is odd,} \\
\sum_{c \in \gf(p)^*}(-1)^{\frac{m}{2}}p^{m/2}\chi_1(-\frac{c}{4})                          & \mbox{ if $m\equiv 2 (\rm{~mod~4})$,} \\
\sum_{c \in \gf(p)^*}-(-1)^{\frac{m}{2}}p^{m/2+1}\chi_1(-\frac{c}{4})                        & \mbox{ if $m\equiv 0 (\rm{~mod~4})$, }
\end{array}
\right.\\
& &=
\left\{ \begin{array}{ll}
 \sqrt{p^*}^m \sum_{c \in \gf(p)^*}\bar{\eta}(c)\bar{\chi}_1(-\frac{c\tr(1)}{4})                               & \mbox{ if $m$ is odd,} \\
-p^{m/2}\sum_{c \in \gf(p)^*}\bar{\chi}_1(\frac{c\tr(1)}{4})                          & \mbox{ if $m\equiv 2 (\rm{~mod~4})$,} \\
-p^{m/2+1}\sum_{c \in \gf(p)^*}\bar{\chi}_1(\frac{c\tr(1)}{4})                        & \mbox{ if $m\equiv 0 (\rm{~mod~4})$ ,}
\end{array}
\right.\\
& &=
\left\{ \begin{array}{ll}
\sqrt{p^*}^m \sum_{c \in \gf(p)^*}\bar{\eta}(c)                          & \mbox{ if $m$ is odd and $m_p=0$,} \\
\sqrt{p^*}^m \sum_{c \in \gf(p)^*}\bar{\eta}(-\frac{m_p}{4})\bar{\eta}(-\frac{cm_p}{4})\bar{\chi}_1(-\frac{cm_p}{4})  & \mbox{ if $m$ is odd and $m_p\neq 0$,} \\
-p^{m/2}(p-1)                          & \mbox{ if $m\equiv 2 (\rm{~mod~4})$ and $m_p=0$,} \\
p^{m/2}                          & \mbox{ if $m\equiv 2 (\rm{~mod~4})$ and $m_p\neq 0$,} \\
-p^{m/2+1}(p-1)                        & \mbox{ if $m\equiv 0 (\rm{~mod~4})$ and $m_p=0$, }\\
p^{m/2+1}                        & \mbox{ if $m\equiv 0 (\rm{~mod~4})$ and $m_p\neq 0$, }
\end{array}
\right.\\
& &=
\left\{ \begin{array}{ll}
0                          & \mbox{ if $m$ is odd and $m_p=0$,} \\
\sqrt{p^*}^m \bar{\eta}(-m_p)G(\bar{\eta}, \bar{\chi}_1)  & \mbox{ if $m$ is odd and $m_p\neq 0$,} \\
-p^{m/2}(p-1)                          & \mbox{ if $m\equiv 2 (\rm{~mod~4})$ and $m_p=0$,} \\
p^{m/2}                          & \mbox{ if $m\equiv 2 (\rm{~mod~4})$ and $m_p\neq 0$,} \\
-p^{m/2+1}(p-1)                        & \mbox{ if $m\equiv 0 (\rm{~mod~4})$ and $m_p=0$, }\\
p^{m/2+1}                        & \mbox{ if $m\equiv 0 (\rm{~mod~4})$ and $m_p\neq 0$ .}
\end{array}
\right.
\end{eqnarray*}
The desired conclusion then follows from Lemma \ref{lem-32A1}.
\end{proof}

The next lemma will be employed in proving the code length.
\begin{lemma}\label{lem-32B2}
Let
$$
n_0=|\{x \in \gf(q): \tr(x^{p+1}-x)=0\}|.
$$
Then
\begin{eqnarray*}
n_0=\left\{ \begin{array}{ll}
p^{m-1}+(-1)^{\frac{p-1}{2}}\bar{\eta}(m_{p})p^{-1}\sqrt{p^*}^{m+1}                               & \mbox{ if $m$ is odd,} \\
p^{m-1}-(p-1)p^{\frac{1}{2}(m-1+(-1)^{\frac{m}{2}})}        & \mbox{ if $m$ is even and $m_p=0$,} \\
p^{m-1}+p^{\frac{1}{2}(m-1+(-1)^{\frac{m}{2}})}              & \mbox{ if $m$ is even and $m_p\neq 0$.}
\end{array}
\right.
\end{eqnarray*}

\end{lemma}
\begin{proof}
By definition, we have
\begin{eqnarray*}
n_0 &=& \frac{1}{p} \sum_{x \in \gf(q)} \sum_{y \in \gf(p)} \zeta_p^{y(\tr(x^{p+1})-x)} \\
&=& p^{m-1}+ \frac{1}{p} \sum_{y \in \gf(p)^*} \sum_{x \in \gf(q)}   \zeta_p^{\tr(yx^{p+1}-y)} \\
&=& p^{m-1}+ \frac{1}{p}\sum_{y \in \gf(p)^*} S(y,-y).
\end{eqnarray*}
The desired conclusion then follows from Lemma \ref{lem-32B1}.
\end{proof}

From Lemma \ref{cyclo}, the conclusion of the following lemma is straightforward and we omit their proofs.
\begin{lemma}\label{lem-A2}With the symbols and notations above, we have the following.
\begin{enumerate}
\item
$
\sum_{y \in \gf(p)^*}\sigma_y(\zeta_p^{~z})
= \left\{ \begin{array}{ll}
p-1     & \mbox{ if $z=0$,} \\
-1     & \mbox{ if $z\neq 0$.}
\end{array}
\right.$

\item

$
\sum_{y \in \gf(p)^*}\sigma_y(\sqrt{p^*}^m)
=\left\{ \begin{array}{ll}
0     & \mbox{ if $m$ is odd,} \\
\sqrt{p^*}^m(p-1)     & \mbox{ if $m$ is even.}
\end{array}
\right.
$
\item
$
\sum_{y \in \gf(p)^*}\sigma_y(\sqrt{p^*}~\zeta_p^{~z})
=\left\{ \begin{array}{ll}
0     & \mbox{ if $z=0$,} \\
\bar{\eta}(z)p^*     & \mbox{ if $z\neq 0$.}
\end{array}
\right.
$
\end{enumerate}
Particularly,
\begin{eqnarray*}
\sum_{y \in \gf(p)^*}\sigma_y(\sqrt{p^*}^m~\zeta_p^{-\frac{m_p}{4}} )
&=&\left\{ \begin{array}{ll}
0     & \mbox{ if $m$ is odd and $m_p=0$,} \\
(-1)^{\frac{p-1}{2}}\bar{\eta}(m_p)\sqrt{p^*}^{m+1}     & \mbox{ if $m$ is odd and $m_p\neq0$.}
\end{array}
\right.
\end{eqnarray*}
and
\begin{eqnarray*}
\sum_{y \in \gf(p)^*}\sigma_y(\zeta_p^{-\frac{m_p}{4}} )
&=&\left\{ \begin{array}{ll}
p-1     & \mbox{ if $m_p=0$,} \\
-1    & \mbox{ if $m_p\neq0$.}
\end{array}
\right.
\end{eqnarray*}
\end{lemma}

The following result will play an important role in proving the main results of this paper.
\begin{lemma}\label{lem-32B3}
Let $b \in \gf(q)^*$, $L(x)=x^{p^2}+x$ and
\begin{eqnarray*}
M=\sum_{y \in \gf(p)^*} \sum_{z \in \gf(p)^*} \sum_{x \in \gf(q)}  \zeta_p^{\tr(yx^{p+1}+(bz-y)x)}.
\end{eqnarray*}

(\uppercase\expandafter{\romannumeral1}) If $m$ is odd, then we have the following.

\begin{itemize}
  \item When $m_p=0$,
\begin{eqnarray*}
M
&=&\left\{ \begin{array}{ll}
0     & \mbox{ if $\tr(x_b^{p+1})=0$,} \\
(-1)^{\frac{p-1}{2}}(p-1)\sqrt{p^*}^{m+1}              & \mbox{ if $\tr(x_b^{p+1}) \in$  $\textup{SQ}$ and $\tr(x_b)=0$,}\\
-(-1)^{\frac{p-1}{2}}(p-1)\sqrt{p^*}^{m+1}   & \mbox{ if $\tr(x_b^{p+1})\in \textup{N\textup{SQ}}$  and $\tr(x_b)=0$,}\\
-(-1)^{\frac{p-1}{2}}\sqrt{p^*}^{m+1}    & \mbox{ if $\tr(x_b^{p+1})\in \textup{SQ}$ and $\tr(x_b)\neq 0$,}\\
(-1)^{\frac{p-1}{2}}\sqrt{p^*}^{m+1} & \mbox{ if $\tr(x_b^{p+1})\in \textup{N\textup{SQ}}$ and $\tr(x_b)\neq 0$.}
\end{array}
\right.
\end{eqnarray*}
  \item When $m_p\in \textup{SQ}$,
\begin{eqnarray*}
M
&=&\left\{ \begin{array}{ll}
(-1)^{\frac{p-1}{2}}(p-1)\sqrt{p^*}^{m+1}     & \mbox{ if $\tr(x_b^{p+1})=0$ and $\tr(x_b)=0$,} \\
-(-1)^{\frac{p-1}{2}}\sqrt{p^*}^{m+1}    & \mbox{ if $\tr(x_b^{p+1})=0$ and $\tr(x_b)\neq 0$,} \\
-2\cdot(-1)^{\frac{p-1}{2}}\sqrt{p^*}^{m+1}            & \mbox{ if $\tr(x_b^{p+1})\in \textup{SQ}$ and $\tr(x_b)=0$,}\\
0  & \mbox{ if $\tr(x_b^{p+1}) \in \textup{N\textup{SQ}} $,}\\
(p-2)(-1)^{\frac{p-1}{2}}\sqrt{p^*}^{m+1}~or~-2\cdot (-1)^{\frac{p-1}{2}}\sqrt{p^*}^{m+1}   & \mbox{ if $\tr(x_b^{p+1})\in \textup{SQ}$ and $\tr(x_b)\neq 0$.}
\end{array}
\right.
\end{eqnarray*}
  \item When $m_p\in \textup{N\textup{SQ}}$,
\begin{eqnarray*}
M
&=&\left\{ \begin{array}{ll}
-(-1)^{\frac{p-1}{2}}(p-1)\sqrt{p^*}^{m+1}     & \mbox{ if $\tr(x_b^{p+1})=0$ and $\tr(x_b)=0$,} \\
(-1)^{\frac{p-1}{2}}\sqrt{p^*}^{m+1}     & \mbox{ if $\tr(x_b^{p+1})=0$ and $\tr(x_b)\neq 0$,} \\
0              & \mbox{ if $\tr(x_b^{p+1})\in \textup{SQ}$,}\\
2\cdot(-1)^{\frac{p-1}{2}}\sqrt{p^*}^{m+1}  & \mbox{ if $\tr(x_b^{p+1})\in \textup{N\textup{SQ}}$ and $\tr(x_b)=0$,}\\
-(p-2)(-1)^{\frac{p-1}{2}}\sqrt{p^*}^{m+1}~or~2\cdot(-1)^{\frac{p-1}{2}}\sqrt{p^*}^{m+1}   & \mbox{ if $\tr(x_b^{p+1})\in \textup{N\textup{SQ}}$ and $\tr(x_b)\neq 0$.}
\end{array}
\right.
\end{eqnarray*}
\end{itemize}
where $x_b$ is the unique solution of the equation $x^{p^2}+x=-b^p$.

(\uppercase\expandafter{\romannumeral2}) If $m\equiv 2 \pmod{4}$, then we have the following.
\begin{itemize}
  \item When $m_p=0$,
  \begin{eqnarray*}
M
&=&\left\{ \begin{array}{ll}
-p^{\frac{m}{2}}(p-1)^2    & \mbox{ if $\tr(x_b^{p+1})=0$ and $\tr(x_b)=0$,} \\
p^{\frac{m}{2}}(p-1)       & \mbox{ if $\tr(x_b^{p+1})=0$ and $\tr(x_b)\neq 0$} \\
~~~~~                       &\mbox{or $\tr(x_b^{p+1})\neq 0$ and $\tr(x_b)= 0$,} \\
-p^{\frac{m}{2}} & \mbox{ if $\tr(x_b^{p+1})\neq 0$ and $\tr(x_b)\neq 0$.}
\end{array}
\right.
\end{eqnarray*}

  \item When $m_p\neq0$,
  \begin{eqnarray*}
M
&=&\left\{ \begin{array}{ll}
p^{\frac{m}{2}}(p-1)   & \mbox{ if $\tr(x_b^{p+1})=0$ and $\tr(x_b)=0$,} \\
-p^{\frac{m}{2}}       & \mbox{ if $\tr(x_b^{p+1})=0$ and $\tr(x_b)\neq 0$}\\
~~~&\mbox{ or $\tr(x_b^{p+1})\neq 0$ and $A= 0$,} \\
p^{\frac{m}{2}}(p-1) ~or ~ -p^{\frac{m}{2}}(p+1) & \mbox{ if $\tr(x_b^{p+1})\neq 0$ and $A\neq 0$,}
\end{array}
\right.
\end{eqnarray*}
\end{itemize}
where $A=-\frac{m_p}{4}+\frac{\tr(x_b)^2}{4\tr(x_b^{p+1})}$ and $x_b$ is the unique solution of the equation $x^{p^2}+x=-b^p$.

(\uppercase\expandafter{\romannumeral3}) If $m \equiv 0 \pmod{4}$, then we have the following.
\begin{itemize}
  \item When $m_p=0$,
\begin{eqnarray*}
M
&=&\left\{ \begin{array}{ll}
0     & \mbox{ if $b\notin \textup{Im}(L)$,} \\
-p^{\frac{m}{2}+1}(p-1)^2    & \mbox{ if $b\in \textup{Im}(L)$, $\tr(x_b^{p+1})=0$ and $\tr(x_b)=0$,} \\
p^{\frac{m}{2}+1}(p-1)       & \mbox{ if $b\in \textup{Im}(L)$, $\tr(x_b^{p+1})=0$ and $\tr(x_b)\neq 0$} \\
~~       & \mbox{~~~ ~~~~~~~~~~or $\tr(x_b^{p+1})\neq 0$ and $\tr(x_b)= 0$,} \\
-p^{\frac{m}{2}+1} & \mbox{ if $b\in \textup{Im}(L)$, $\tr(x_b^{p+1})\neq 0$ and $\tr(x_b)\neq 0$.}
\end{array}
\right.
\end{eqnarray*}

\item
When $m_p\neq0$,
\begin{eqnarray*}
M
&=&\left\{ \begin{array}{ll}
0     & \mbox{ if $b\notin \textup{Im}(L)$,} \\
p^{\frac{m}{2}+1}(p-1)    & \mbox{ if $b\in \textup{Im}(L)$, $\tr(x_b^{p+1})=0$ and $\tr(x_b)=0$,} \\
-p^{\frac{m}{2}+1}       & \mbox{ if $b\in \textup{Im}(L)$, $\tr(x_b^{p+1})=0$ and $\tr(x_b)\neq 0$} \\
~~       & \mbox{~~~ ~~~~~~~~~~or $\tr(x_b^{p+1})\neq 0$ and $A=0$,} \\
p^{\frac{m}{2}+1}(p-1) ~or~ -p^{\frac{m}{2}+1}(p+1)  & \mbox{ if $b\in \textup{Im}(L)$, $\tr(x_b^{p+1})\neq 0$ and $A\neq 0$.}
\end{array}
\right.
\end{eqnarray*}
\end{itemize}
where $A=-\frac{m_p}{4}+\frac{\tr(x_b)^2}{4\tr(x_b^{p+1})}$ and $x_b$ is some solution of the equation $L(x)=-b^p$ under the condition that $b\in \textup{Im}(L)$.
\end{lemma}

\begin{proof}
We have
\begin{eqnarray}\label{eqn-M1}
M &=& \sum_{y \in \gf(p)^*} \sum_{z \in \gf(p)^*} \sum_{x \in \gf(q)}  \zeta_p^{\tr(yx^{p+1}+(bz-y)x)} \nonumber \\
&=& \sum_{y \in \gf(p)^*} \sum_{z \in \gf(p)^*} \sum_{x \in \gf(q)}  \zeta_p^{y\tr(x^{p+1}+(\frac{z}{y}b-1)x)}   \nonumber \\
&=& \sum_{y \in \gf(p)^*} \sum_{z \in \gf(p)^*} \sum_{x \in \gf(q)}  \zeta_p^{y\tr(x^{p+1}+(zb-1)x)}   \nonumber \\
&=& \sum_{y \in \gf(p)^*} \sigma_y(\sum_{z \in \gf(p)^*} \sum_{x \in \gf(q)}  \zeta_p^{\tr(x^{p+1}+(bz-1)x)}).
\end{eqnarray}

 It is easily seen that $L(x)=x^{p^2}+x$ is a permutation polynomial over $\gf(q)$ if $m\equiv 2 \pmod{4}$ or $m$ is odd. Hence, $x_b$ is the unique solution of the equation $L(x)=-b^p$, while $zx_b+\frac{1}{2}$ is the unique solution of the equation $L(x)=-(bz-1)^p$ for $z\in \gf(p)^*$.

For $m\equiv 0 \pmod{4}$, although $L(x)=x^{p^2}+x$ is not a permutation polynomial over $\gf(q)$, $zx_b+\frac{1}{2}$ is the solution of the equation $L(x)=-(bz-1)^p$ for $z\in \gf(p)^*$ when $b\in \textup{Im}(L)$ with some solution $x_b$.

Therefore, by Lemmas \ref{lem-Coulter1} and \ref{lem-Coulter2}, Equation (\ref{eqn-M1}) becomes

\begin{eqnarray*}
M
&=&\left\{ \begin{array}{ll}
     \sum_{y \in \gf(p)^*} \sigma_y(\sum_{z \in \gf(p)^*} \sqrt{p^*}^m~\eta(1)\zeta_p^{-\tr((zx_b+\frac{1}{2})^{p+1})}) & \mbox{ if $m$ is odd} \\
\sum_{y \in \gf(p)^*} \sigma_y(\sum_{z \in \gf(p)^*} (-1)^{\frac{m}{2}}p^{\frac{m}{2}}\zeta_p^{-\tr((zx_b+\frac{1}{2})^{p+1})}) & \mbox{ if $m\equiv2 (\rm{~mod~}4)$}\\
0                        & \mbox{ if $m\equiv0 (\rm{~mod~}4)$ and $b\notin \textup{Im}(L)$}\\
\sum_{y \in \gf(p)^*} \sigma_y(\sum_{z \in \gf(p)^*} -(-1)^{\frac{m}{2}}p^{\frac{m}{2}+1}\zeta_p^{-\tr((zx_b+\frac{1}{2})^{p+1})}) & \mbox{ if $m\equiv0 (\rm{~mod~}4)$ and $b\in \textup{Im}(L)$}
\end{array}
\right.\\
&=&\left\{ \begin{array}{ll}
\sum_{y \in \gf(p)^*} \sigma_y(\sqrt{p^*}^m~\zeta_p^{-\frac{m_p}{4}}\sum_{z \in \gf(p)^*} \zeta_p^{-\tr(x_b^{p+1})z^2-\tr(x_b)z})  & \mbox{ if $m$ is odd,} \\
\sum_{y \in \gf(p)^*} \sigma_y(-p^{\frac{m}{2}}~\zeta_p^{-\frac{m_p}{4}}\sum_{z \in \gf(p)^*} \zeta_p^{-\tr(x_b^{p+1})z^2-\tr(x_b)z}) & \mbox{ if $m\equiv2 (\rm{~mod~}4)$,}\\
0                                                     & \mbox{ if $m\equiv0 (\rm{~mod~}4)$ and $b\notin \textup{Im}(L)$,}\\
\sum_{y \in \gf(p)^*} \sigma_y(-p^{\frac{m}{2}+1}~\zeta_p^{-\frac{m_p}{4}}\sum_{z \in \gf(p)^*} \zeta_p^{-\tr(x_b^{p+1})z^2-\tr(x_b)z}) & \mbox{ if $m\equiv0 (\rm{~mod~}4)$and $b\in \textup{Im}(L)$,}
\end{array}
\right.
\end{eqnarray*}
We distinguish the following three cases.

(\uppercase\expandafter{\romannumeral1}) If $m$ is odd, by Lemma \ref{lem-32A2}, we have
\begin{eqnarray*}
M
&=&\left\{ \begin{array}{ll}
\sum_{y \in \gf(p)^*} \sigma_y(\sqrt{p^*}^m~\zeta_p^{-\frac{m_p}{4}}\sum_{z \in \gf(p)^*} \zeta_p^{-\tr(x_b)z})     & \mbox{ if $\tr(x_b^{p+1})=0$} \\
\sum_{y \in \gf(p)^*} \sigma_y(\sqrt{p^*}^m~\zeta_p^{-\frac{m_p}{4}}(\bar{\eta}(-\tr(x_b^{p+1})\zeta_p^{\frac{\tr(x_b)^2}{4\tr(x_b^{p+1})}}G(\bar{\eta},\bar{\chi}_1)-1))    & \mbox{ if $\tr(x_b^{p+1})\neq 0$}
\end{array}
\right.\\
&=&\left\{ \begin{array}{ll}
(p-1)A_0     & \mbox{ if $\tr(x_b^{p+1})=0$ and $\tr(x_b)=0$,} \\
-A_0     & \mbox{ if $\tr(x_b^{p+1})=0$ and $\tr(x_b)\neq 0$,} \\
\sqrt{p^*}^{m+1}\bar{\eta}(-\tr(x_b^{p+1})\sum_{y \in \gf(p)^*} \sigma_y(\zeta_p^{-\frac{m_p}{4}})-A_0             & \mbox{ if $\tr(x_b^{p+1})\neq 0$ and $\tr(x_b)=0$,}\\
\sqrt{p^*}^{m+1}\bar{\eta}(-\tr(x_b^{p+1})\sum_{y \in \gf(p)^*} \sigma_y(\zeta_p^{~A})-A_0               & \mbox{ if $\tr(x_b^{p+1})\neq 0$ and $\tr(x_b)\neq 0$,}
\end{array}
\right.
\end{eqnarray*}
where $A_0=\sum_{y \in \gf(p)^*} \sigma_y(\sqrt{p^*}^m~\zeta_p^{-\frac{m_p}{4}})$ and $A=-\frac{m_p}{4}+\frac{\tr(x_b)^2}{4\tr(x_b^{p+1})}$.
The desired conclusion in Part (\uppercase\expandafter{\romannumeral1}) of this lemma then follows from Lemma \ref{lem-A2}.

(\uppercase\expandafter{\romannumeral2}) If $m\equiv  2 \pmod{4}$, by Lemmas \ref{lem-32A2} and \ref{lem-32A1}, we have
\begin{eqnarray*}
M &=&
\left\{ \begin{array}{ll}
\sum_{y \in \gf(p)^*} \sigma_y(-p^{\frac{m}{2}}~\zeta_p^{-\frac{m_p}{4}}\sum_{z \in \gf(p)^*} \zeta_p^{-\tr(x_b)z})     & \mbox{ if $\tr(x_b^{p+1})=0$} \\
\sum_{y \in \gf(p)^*} \sigma_y(-p^{\frac{m}{2}}~\zeta_p^{-\frac{m_p}{4}}(\sum_{z \in \gf(p)} \zeta_p^{-\tr(x_b^{p+1})z^2-\tr(x_b)z}-1))              & \mbox{ if $\tr(x_b^{p+1})\neq 0$}
\end{array}
\right.\\
&=&\left\{ \begin{array}{ll}
-p^{\frac{m}{2}}\sum_{y \in \gf(p)^*} \sigma_y(~\zeta_p^{-\frac{m_p}{4}}\sum_{z \in \gf(p)^*} \zeta_p^{-\tr(x_b)z})     & \mbox{ if $\tr(x_b^{p+1})=0$} \\
-p^{\frac{m}{2}}\sum_{y \in \gf(p)^*} \sigma_y(~\zeta_p^{-\frac{m_p}{4}}(\bar{\eta}(-\tr(x_b^{p+1})\zeta_p^{\frac{\tr(x_b)^2}{4\tr(x_b^{p+1})}}G(\bar{\eta},\bar{\chi}_1)-1))              & \mbox{ if $\tr(x_b^{p+1})\neq 0$}
\end{array}
\right.\\
&=&\left\{ \begin{array}{ll}
-p^{\frac{m}{2}}(p-1)A_1    & \mbox{ if $\tr(x_b^{p+1})=0$ and $\tr(x_b)=0$,} \\
p^{\frac{m}{2}}A_1     & \mbox{ if $\tr(x_b^{p+1})=0$ and $\tr(x_b)\neq 0$,} \\
-p^{\frac{m}{2}}(-1)^{\frac{p-1}{2}}\bar{\eta}(\tr(x_b^{p+1}))\sum_{y \in \gf(p)^*} \sigma_y(\sqrt{p^*}\zeta_p^{-\frac{m_p}{4}})+p^{\frac{m}{2}}A_1              & \mbox{ if $\tr(x_b^{p+1})\neq 0$ and $\tr(x_b)=0$,}\\
-p^{\frac{m}{2}}(-1)^{\frac{p-1}{2}}\bar{\eta}(\tr(x_b^{p+1}))\sum_{y \in \gf(p)^*} \sigma_y(\sqrt{p^*}\zeta_p^{~A})+p^{\frac{m}{2}}A_1               & \mbox{ if $\tr(x_b^{p+1})\neq 0$ and $\tr(x_b)\neq 0$,}
\end{array}
\right.
\end{eqnarray*}
where $A_1=\sum_{y \in \gf(p)^*} \sigma_y(~\zeta_p^{-\frac{m_p}{4}})$ and $A=-\frac{m_p}{4}+\frac{\tr(x_b)^2}{4\tr(x_b^{p+1})}$.
The desired conclusion in Part (\uppercase\expandafter{\romannumeral2}) of this lemma then follows from Lemma \ref{lem-A2}.

(\uppercase\expandafter{\romannumeral3}) If $m\equiv  0 \pmod{4}$, then from Lemmas \ref{lem-32A2} and \ref{lem-32A1} we have

\begin{eqnarray*}
M &=&
\left\{ \begin{array}{ll}
0
~~~~~~~~~~~~~~~~~~~~~~~~~~~~~
~~~~~~~~~~~~~~~~~~~~~~~~~~~~~
~~~~~~~~~~~~~~~~~~~~~~~~~~~~~~~~~\mbox{ if $b\notin \textup{Im}(L)$}                        & \\
-p^{\frac{m}{2}+1}\sum_{y \in \gf(p)^*} \sigma_y(\zeta_p^{-\frac{m_p}{4}}\sum_{z \in \gf(p)^*} \zeta_p^{-\tr(x_b)z})
~~~~~~~~~~~~~~~~~~~~~~~\mbox{ if $b\in \textup{Im}(L)$ and $\tr(x_b^{p+1})=0$}                             &\\
-p^{\frac{m}{2}+1}\sum_{y \in \gf(p)^*} \sigma_y(\zeta_p^{-\frac{m_p}{4}}(\sum_{z \in \gf(p)} \zeta_p^{-\tr(x_b^{p+1})z^2-\tr(x_b)z}-1))
~\mbox{ if $b\in \textup{Im}(L)$ and $\tr(x_b^{p+1})\neq 0$} & \\
\end{array}
\right.\\
&=&\left\{ \begin{array}{ll}
0 &\mbox{ if $b\notin \textup{Im}(L)$,} \\
-p^{\frac{m}{2}+1}(p-1)A_1    &\mbox{ if $b\in \textup{Im}(L)$, $\tr(x_b^{p+1})=0$ and $\tr(x_b)=0$,} \\
p^{\frac{m}{2}+1}A_1    &\mbox{ if $b\in \textup{Im}(L)$, $\tr(x_b^{p+1})=0$ and $\tr(x_b)\neq 0$,} \\
-p^{\frac{m}{2}+1}
(-1)^{\frac{p-1}{2}}
\bar{\eta}(\tr(x_b^{p+1}))A
+p^{\frac{m}{2}+1}A_1
&\mbox{ if  $b\in \textup{Im}(L)$ and $\tr(x_b^{p+1})\neq 0$,}
\end{array}
\right.
\end{eqnarray*}
where $A_1=\sum_{y \in \gf(p)^*} \sigma_y(~\zeta_p^{-\frac{m_p}{4}})$ and
$A=\sum_{y \in \gf(p)^*} \sigma_y(\sqrt{p^*}\zeta_p^{~-\frac{m_p}{4}+
\frac{\tr(x_b)^2}{4\tr(x_b^{p+1})}})$.
The desired conclusion in Part (\uppercase\expandafter{\romannumeral3}) of this lemma then follows from Lemma \ref{lem-A2}.

This completes the proof of this lemma.
\end{proof}

In order to calculate the Hamming weight of each codeword in $\C_D$, we need the following result.

\begin{lemma}\label{lem-32B4}
For any $b \in \gf(q)^*$, let
$$
N(b)=|\{x \in \gf(q): \tr(x^{p+1}-x)=0 \mbox{ and } \tr(bx)=0\}|.
$$
There are three cases.

(\uppercase\expandafter{\romannumeral1}) If $m$ is odd, then we have the following.
\begin{description}
\item[$\bullet$] When $m_p=0$,
\begin{eqnarray*}
N(b)
&=&\left\{ \begin{array}{ll}
p^{m-2}     & \mbox{ if $\tr(x_b^{p+1})=0$,} \\
(p-1)B+p^{m-2}              & \mbox{ if $\tr(x_b^{p+1}) \in \textup{SQ}$ and $\tr(x_b)=0$,}\\
-(p-1)B+p^{m-2}   & \mbox{ if $\tr(x_b^{p+1})\in \textup{N\textup{SQ}}$ and $\tr(x_b)=0$,}\\
-B+p^{m-2}    & \mbox{ if $\tr(x_b^{p+1})\in \textup{SQ}$ and $\tr(x_b)\neq 0$,}\\
B+p^{m-2}   & \mbox{ if $\tr(x_b^{p+1})\in \textup{N\textup{SQ}}$ and $\tr(x_b)\neq 0$.}
\end{array}
\right.
\end{eqnarray*}
\item[$\bullet$]  When $\bar{\eta}(m_p)=1$,
\begin{eqnarray*}
N(b)
&=&\left\{ \begin{array}{ll}
pB+{p}^{m-2}     & \mbox{ if $\tr(x_b^{p+1})=0$ and $\tr(x_b)=0$,} \\
{p}^{m-2}              & \mbox{ if $\tr(x_b^{p+1})=0$ and $\tr(x_b)\neq 0$,} \\
-B+{p}^{m-2}            & \mbox{ if $\tr(x_b^{p+1})\in \textup{SQ}$ and $\tr(x_b)=0$,}\\
B+{p}^{m-2}             & \mbox{ if $\tr(x_b^{p+1}) \in \textup{N\textup{SQ}}$,}\\
(p-1)B+{p}^{m-2}~or~-B+{p}^{m-2}   & \mbox{ if $\tr(x_b^{p+1})\in \textup{SQ}$ and $\tr(x_b)\neq 0$.}
\end{array}
\right.
\end{eqnarray*}
\item[$\bullet$]
When $\bar{\eta}(m_p)=-1$,
\begin{eqnarray*}
N(b)
&=&\left\{ \begin{array}{ll}
-p B+ {p}^{m-2}     & \mbox{ if $\tr(x_b^{p+1})=0$ and $\tr(x_b)=0$,} \\
{p}^{m-2}     & \mbox{ if $\tr(x_b^{p+1})=0$ and $\tr(x_b)\neq 0$,} \\
-B+{p}^{m-2}              & \mbox{ if $\tr(x_b^{p+1})\in \textup{SQ}$,}\\
B+{p}^{m-2}  & \mbox{ if $\tr(x_b^{p+1})\in \textup{N\textup{SQ}}$ and $\tr(x_b)=0$,}\\
-(p-1)B+{p}^{m-2}~or~B+{p}^{m-2}   & \mbox{ if $\tr(x_b^{p+1})\in \textup{N\textup{SQ}}$ and $\tr(x_b)\neq 0$.}
\end{array}
\right.
\end{eqnarray*}
\end{description}
where $B=(-1)^{\frac{p-1}{2}\frac{m-1}{2}}p^{\frac{m-3}{2}}$.

(\uppercase\expandafter{\romannumeral2}) If $m\equiv 2 \pmod{4}$, then we have the following.
\begin{description}
  \item[$\bullet$] When $m_p=0$,
  \begin{eqnarray*}
N(b)
&=&\left\{ \begin{array}{ll}
-(p-1)p^{\frac{m}{2}-1}+p^{m-2}    & \mbox{ if $\tr(x_b^{p+1})=0$ and $\tr(x_b)=0$,} \\
p^{m-2}       & \mbox{ if $\tr(x_b^{p+1})=0$ and $\tr(x_b)\neq 0$,} \\
~~            & \mbox{or $\tr(x_b^{p+1})\neq 0$ and $\tr(x_b)= 0$,} \\
-p^{\frac{m}{2}-1}+p^{m-2} & \mbox{ if $\tr(x_b^{p+1})\neq 0$ and $\tr(x_b)\neq 0$.}
\end{array}
\right.
\end{eqnarray*}

   \item[$\bullet$] When $m_p\neq0$,
  \begin{eqnarray*}
N(b)
&=&\left\{ \begin{array}{ll}
p^{\frac{m}{2}-1}+p^{m-2}   & \mbox{ if $\tr(x_b^{p+1})=0$ and $\tr(x_b)=0$,} \\
p^{m-2}      & \mbox{ if $\tr(x_b^{p+1})=0$ and $\tr(x_b)\neq 0$,} \\
~~           & \mbox{or $\tr(x_b^{p+1})\neq 0$ and $A= 0$,} \\
p^{\frac{m}{2}-1}+p^{m-2} ~or ~ -p^{\frac{m}{2}-1}+p^{m-2} & \mbox{ if $\tr(x_b^{p+1})\neq 0$ and $A\neq 0$,}
\end{array}
\right.
\end{eqnarray*}
\end{description}
where $A=-\frac{m_p}{4}+\frac{\tr(x_b)^2}{4\tr(x_b^{p+1})}$ and $x_b$ is the unique solution of the equation $x^{p^2}+x=-b^p$.

(\uppercase\expandafter{\romannumeral3}) If $m\equiv 0 \pmod{4}$, then we have the following.
\begin{description}
\item[$\bullet$] When $m_p=0$,
\begin{eqnarray*}
N(b)
&=&\left\{ \begin{array}{ll}
-(p-1)p^{\frac{m}{2}-1}+p^{m-2}     & \mbox{ if $b\notin \textup{Im}(L)$,} \\
-(p-1)p^{\frac{m}{2}}+p^{m-2}    & \mbox{ if $b\in \textup{Im}(L)$, $\tr(x_b^{p+1})=0$ and $\tr(x_b)=0$,} \\
p^{m-2}       & \mbox{ if $b\in \textup{Im}(L)$, $\tr(x_b^{p+1})=0$ and $\tr(x_b)\neq 0$,} \\
~~       & \mbox{~~~ ~~~~~~~~~~or $\tr(x_b^{p+1})\neq 0$ and $\tr(x_b)= 0$,} \\
-p^{\frac{m}{2}}+p^{m-2} & \mbox{ if $b\in \textup{Im}(L)$, $\tr(x_b^{p+1})\neq 0$ and $\tr(x_b)\neq 0$.}
\end{array}
\right.
\end{eqnarray*}

\item[$\bullet$]
When $m_p\neq0$,
\begin{eqnarray*}
N(b)
&=&\left\{ \begin{array}{ll}
p^{\frac{m}{2}-1}+p^{m-2}     & \mbox{ if $b\notin \textup{Im}(L)$,} \\
p^{\frac{m}{2}}+p^{m-2}    & \mbox{ if $b\in \textup{Im}(L)$, $\tr(x_b^{p+1})=0$ and $\tr(x_b)=0$,} \\
p^{m-2}       & \mbox{ if $b\in \textup{Im}(L)$, $\tr(x_b^{p+1})=0$ and $\tr(x_b)\neq 0$,} \\
~~       & \mbox{~~~~~~~~~~~~~ or $\tr(x_b^{p+1})\neq 0$ and $A=0$,} \\
p^{\frac{m}{2}}+p^{m-2} ~or ~ -p^{\frac{m}{2}}+p^{m-2}  & \mbox{ if $b\in \textup{Im}(L)$, $\tr(x_b^{p+1})\neq 0$ and $A\neq 0$.}
\end{array}
\right.
\end{eqnarray*}
\end{description}
where $A=-\frac{m_p}{4}+\frac{\tr(x_b)^2}{4\tr(x_b^{p+1})}$ and $x_b$ is some solution of the equation $L(x)=x^{p^2}+x=-b^p$ when $b\in \textup{Im}(L)$.
\end{lemma}

\begin{proof}
By definition, we have
\begin{eqnarray*}
N(b)
&=& p^{-2} \sum_{x \in \gf(q)} \left( \sum_{y \in \gf(p)} \zeta_p^{y\tr(x^{p+1}-x)} \right)
                                                   \left( \sum_{z \in \gf(p)} \zeta_p^{z\tr(bx)} \right) \\
&=& p^{-2}  \sum_{z \in \gf(p)^*} \sum_{x \in \gf(q)}  \zeta_p^{\tr(bzx)}  +
        p^{-2}  \sum_{y \in \gf(p)^*} \sum_{x \in \gf(q)}  \zeta_p^{\tr(yx^{p+1}-yx)} + \\
& & p^{-2}   \sum_{y \in \gf(p)^*}  \sum_{z \in \gf(p)^*} \sum_{x \in \gf(q)}  \zeta_p^{\tr(yx^{p+1}+(bz-y)x)} + p^{m-2}.
\end{eqnarray*}
Note that
$$
\sum_{z \in \gf(p)^*} \sum_{x \in \gf(q)}  \zeta_p^{\tr(bzx)} =0.
$$
The desired conclusions then follow from Lemmas \ref{lem-32B1} and \ref{lem-32B3}, and the fact that
\begin{eqnarray*}
(-1)^{\frac{p-1}{2}}\sqrt{p^*}^{m+1}=(-1)^{{\frac{p-1}{2}\frac{m-1}{2}}}p^{\frac{m+1}{2}}.
\end{eqnarray*}
\end{proof}

In order to calculate the frequency of each weight in $\C_D$, we need a few more auxiliary results which are given and proved in the following three lemmas.

\begin{lemma}\label{lem-F1}
For any $a \in \gf(p)$, let
$$
N_a=|\{x \in \gf(q): \tr(x^{p+1})=a\}|.
$$
Then
\begin{eqnarray*}
N_a = \left\{ \begin{array}{ll}
p^{m-1}                               & \mbox{ if $m$ is odd and $a=0$,} \\
p^{m-1}+p^{-1}\sqrt{p^*}^{m+1} \bar{\eta}(-a)                              & \mbox{ if $m$ is odd and $a\neq 0$,} \\
p^{m-1}-p^{\frac{m}{2}-1}(p-1)        & \mbox{ if $m\equiv 2 (\rm{~mod~}4)$ and $a=0$,} \\
p^{m-1}+p^{\frac{m}{2}-1}        & \mbox{ if $m\equiv 2 (\rm{~mod~}4)$ and $a\neq 0$,} \\
p^{m-1}-p^{\frac{m}{2}}(p-1)             & \mbox{ if $m\equiv0 (\rm{~mod~}4)$ and $a=0$,}\\
p^{m-1}+p^{\frac{m}{2}}            & \mbox{ if $m\equiv0 (\rm{~mod~}4)$ and $a \neq 0$.}
\end{array}
\right.
\end{eqnarray*}
\end{lemma}
\begin{proof}
By definition, Lemmas \ref{lem-Coulter1} and \ref{lem-Coulter2}, we have
\begin{eqnarray*}
N_a &=& \frac{1}{p} \sum_{x \in \gf(q)} \sum_{y \in \gf(p)} \zeta_p^{y(\tr(x^{p+1})-a)} \\
&=& p^{m-1}+ \frac{1}{p} \sum_{y \in \gf(p)^*} \zeta_p^{-ya}\sum_{x \in \gf(q)}   \zeta_p^{\tr(yx^{p+1})} \\
&=& p^{m-1}+ \frac{1}{p}\sum_{y \in \gf(p)^*} \zeta_p^{-ya} S(y,0)\\
&=& \left\{ \begin{array}{ll}
p^{m-1}+p^{-1}\sum_{y \in \gf(p)^*} \zeta_p^{-ya}\sqrt{p^*}^m \eta(y)                               & \mbox{ if $m$ is odd} \\
p^{m-1}+p^{-1}\sum_{y \in \gf(p)^*} \zeta_p^{-ya}(-p^{\frac{m}{2}})        & \mbox{ if $m\equiv2 (\rm{~mod~}4)$} \\
p^{m-1}+p^{-1}\sum_{y \in \gf(p)^*} \zeta_p^{-ya}(-p^{\frac{m}{2}+1})             & \mbox{ if $m\equiv0 (\rm{~mod~}4)$}
\end{array}
\right.\\
&=& \left\{ \begin{array}{ll}
p^{m-1}+p^{-1}\sqrt{p^*}^m \sum_{y \in \gf(p)^*} \bar{\eta}(y)                               & \mbox{ if $m$ is odd and $a=0$} \\
p^{m-1}+p^{-1}\sqrt{p^*}^m \bar{\eta}(-a)\sum_{y \in \gf(p)^*} \zeta_p^{-ya}\bar{\eta}(-ya)                               & \mbox{ if $m$ is odd and $a\neq 0$} \\
p^{m-1}+p^{-1}(-p^{\frac{m}{2}})\sum_{y \in \gf(p)^*} \zeta_p^{-ya}        & \mbox{ if $m\equiv2 (\rm{~mod~}4)$} \\
p^{m-1}+p^{-1}(-p^{\frac{m}{2}+1})\sum_{y \in \gf(p)^*} \zeta_p^{-ya}             & \mbox{ if $m\equiv0 (\rm{~mod~}4)$}
\end{array}
\right.\\
&=& \left\{ \begin{array}{ll}
p^{m-1}+p^{-1}\sqrt{p^*}^m \sum_{y \in \gf(p)^*} \bar{\eta}(y)                               & \mbox{ if $m$ is odd and $a=0$,} \\
p^{m-1}+p^{-1}\sqrt{p^*}^m \bar{\eta}(-a)G(\bar{\eta},\bar{\chi}_1)                               & \mbox{ if $m$ is odd and $a\neq 0$,} ~~~~~~~~~~~~\\
p^{m-1}-p^{\frac{m}{2}-1}\sum_{y \in \gf(p)^*} \zeta_p^{-ya}        & \mbox{ if $m\equiv 2(\rm{~mod~}4)$,} \\
p^{m-1}-p^{\frac{m}{2}}\sum_{y \in \gf(p)^*} \zeta_p^{-ya}             & \mbox{ if $m\equiv0 (\rm{~mod~}4)$.}
\end{array}
\right.
\end{eqnarray*}

Note that $\sum_{y \in \gf(p)^*} \bar{\eta}(y)=0$ and
\begin{eqnarray*}
\sum_{y \in \gf(p)^*} \zeta_p^{-ya} &=& \left\{ \begin{array}{ll}
p-1                               & \mbox{ if $a=0$,} \\
-1                               & \mbox{ if $a\neq 0$,}
\end{array}
\right.
\end{eqnarray*}

The desired conclusion then follows from Lemma \ref{lem-32A1}.
\end{proof}

\begin{lemma}\label{lem-F2}
For any $a\in \gf(p)$, let
$$
N_{(a,0)}=|\{x \in \gf(q): \tr(x^{p+1})=a ~and~ \tr(x)=0\}|.
$$
Then, for $m_p=0$, we have
\begin{eqnarray*}
N_{(a,0)}
&=&
\left\{ \begin{array}{ll}
p^{m-2}                                 & \mbox{ if $m$ is odd and $a=0$,} \\
p^{m-2} +
p^{-1}\sqrt{p^*}^{m+1}\bar{\eta}(-a)    & \mbox{ if $m$ is odd and $a\neq 0$,} \\
p^{m-2}
-p^{\frac{m}{2}-1}(p-1)       & \mbox{ if $m\equiv 2 (\rm{~mod~}4)$ and $a=0$,} \\
p^{m-2}
+p^{\frac{m}{2}-1}       & \mbox{ if $m\equiv 2 (\rm{~mod~}4)$ and $a\neq0$,} \\
p^{m-2}
-p^{\frac{m}{2}}(p-1)       & \mbox{ if $m\equiv0 (\rm{~mod~}4)$ and $a=0$,} \\
p^{m-2}
+p^{\frac{m}{2}}       & \mbox{ if $m\equiv0 (\rm{~mod~}4)$ and $a\neq0$.}
\end{array}
\right.
\end{eqnarray*}
and for $m_p \neq 0$, we have
\begin{eqnarray*}
N_{(a,0)} &=& \left\{ \begin{array}{ll}
p^{m-2}+p^{-2}\bar{\eta}(-m_p)\sqrt{p^*}^{m+1}(p-1)                               & \mbox{ if $m$ is odd and $a=0$,} \\
p^{m-2}-p^{-2}\bar{\eta}(-m_p)\sqrt{p^*}^{m+1}                               & \mbox{ if $m$ is odd and $a\neq0$,} \\
p^{m-2}                                                             & \mbox{ if $m$ is even and $a=0$,}\\
p^{m-2}
-\bar{\eta}(m_p)\bar{\eta}(a)(-1)^{\frac{p-1}{2}}p^{\frac{m}{2}-1}       & \mbox{ if $m\equiv 2 (\rm{~mod~}4)$ and $a\neq0$,} \\
p^{m-2}
-\bar{\eta}(m_p)\bar{\eta}(a)(-1)^{\frac{p-1}{2}}p^{\frac{m}{2}}  & \mbox{ if $m\equiv0 (\rm{~mod~}4)$ and $a\neq0$.}
\end{array}
\right.
\end{eqnarray*}
\end{lemma}

\begin{proof}
By definition, we have
\begin{eqnarray*}
N_{(a,0)}
&=& p^{-2} \sum_{x \in \gf(q)} \left( \sum_{y \in \gf(p)} \zeta_p^{y(\tr(x^{p+1})-a)} \right)
                                                   \left( \sum_{z \in \gf(p)} \zeta_p^{z\tr(x)} \right) \\
&=& p^{-2}  \sum_{z \in \gf(p)} \sum_{x \in \gf(q)}  \zeta_p^{\tr(zx)}  +
p^{-2}   \sum_{y \in \gf(p)^*}  \sum_{z \in \gf(p)} \sum_{x \in \gf(q)}  \zeta_p^{\tr(yx^{p+1}+zx-ya)}\\
&=& p^{-2} ( q+\sum_{z \in \gf(p)^*} \sum_{x \in \gf(q)}  \zeta_p^{\tr(zx)})  +
p^{-2}   \sum_{y \in \gf(p)^*}  \left(\zeta_p^{-ya}\sum_{z \in \gf(p)} \sum_{x \in \gf(q)}  \zeta_p^{y(\tr(x^{p+1}+\frac{z}{y}x))} \right).
\end{eqnarray*}
Note that
$$
\sum_{z \in \gf(p)^*} \sum_{x \in \gf(q)}  \zeta_p^{\tr(zx)} =0.
$$
Therefore,
\begin{eqnarray}\label{eqn-a0}
N_{(a,0)}
&=& p^{m-2}+
p^{-2}   \sum_{y \in \gf(p)^*}  \left(\zeta_p^{-ya}\sum_{z \in \gf(p)} \sum_{x \in \gf(q)}  \zeta_p^{y(\tr(x^{p+1}+zx))} \right) \nonumber \\
&=& p^{m-2} +
p^{-2}   \sum_{y \in \gf(p)^*}  \sigma_y \left(\zeta_p^{-a}\sum_{z \in \gf(p)} \sum_{x \in \gf(q)}  \zeta_p^{\tr(x^{p+1}+zx)} \right).
\end{eqnarray}
It is clear that $-\frac{1}{2}z$ is the solution of the equation $x^{p^2}+x=-z^p$ for any $z\in \gf(p)$. Hence, by Lemmas \ref{lem-Coulter1} and \ref{lem-Coulter2}, Equation (\ref{eqn-a0}) becomes
\begin{eqnarray}\label{eqn-a02}
N_{(a,0)}
&=& p^{m-2} +
p^{-2}   \sum_{y \in \gf(p)^*}  \sigma_y \left(\zeta_p^{-a}\sum_{z \in \gf(p)} \sum_{x \in \gf(q)}  \zeta_p^{\tr(x^{p+1}+zx)} \right) \nonumber \\
&=&
\left\{ \begin{array}{ll}
p^{m-2} +
p^{-2}\sum_{y \in \gf(p)^*} \sigma_y(\zeta_p^{-a}\sum_{z \in \gf(p)} \sqrt{p^*}^m \zeta_p^{-\frac{m_p}{4}z^2})    & \mbox{if $m$ is odd,} \\
p^{m-2} +
p^{-2}\sum_{y \in \gf(p)^*} \sigma_y(\zeta_p^{-a}\sum_{z \in \gf(p)} -p^{\frac{m}{2}} \zeta_p^{-\frac{m_p}{4}z^2})       & \mbox{if $m\equiv2 (\rm{~mod~}4)$,} \\
p^{m-2} +
p^{-2}\sum_{y \in \gf(p)^*} \sigma_y(\zeta_p^{-a}\sum_{z \in \gf(p)} -p^{\frac{m}{2}+1} \zeta_p^{-\frac{m_p}{4}z^2})          & \mbox{if $m\equiv0 (\rm{~mod~}4)$.}
\end{array}
\right.
\end{eqnarray}
We distinguish the following two cases.
\begin{enumerate}
  \item When $m_p=0$, Equation (\ref{eqn-a02}) becomes
\begin{eqnarray*}
N_{(a,0)}
&=&
\left\{ \begin{array}{ll}
p^{m-2} +
p^{-2}\sum_{y \in \gf(p)^*} \sigma_y(\zeta_p^{-a}\sum_{z \in \gf(p)} \sqrt{p^*}^m)    & \mbox{if $m$ is odd,} \\
p^{m-2} +
p^{-2}\sum_{y \in \gf(p)^*} \sigma_y(\zeta_p^{-a}\sum_{z \in \gf(p)} -p^{\frac{m}{2}})       & \mbox{if $m\equiv2 (\rm{~mod~}4)$,} \\
p^{m-2} +
p^{-2}\sum_{y \in \gf(p)^*} \sigma_y(\zeta_p^{-a}\sum_{z \in \gf(p)} -p^{\frac{m}{2}+1})          & \mbox{if $m\equiv0 (\rm{~mod~}4)$,}
\end{array}
\right.\\
&=&
\left\{ \begin{array}{ll}
p^{m-2} +
p^{-1}\sqrt{p^*}^{m-1}\sum_{y \in \gf(p)^*} \sigma_y(\zeta_p^{-a}\sqrt{p^*})    & \mbox{ if $m$ is odd,} \\
p^{m-2}
-p^{\frac{m}{2}-1}\sum_{y \in \gf(p)^*} \sigma_y(\zeta_p^{-a})       & \mbox{ if $m\equiv2 (\rm{~mod~}4)$,} \\
p^{m-2}
-p^{\frac{m}{2}}\sum_{y \in \gf(p)^*} \sigma_y(\zeta_p^{-a})          & \mbox{ if $m \equiv0 (\rm{~mod~}4)$.}
\end{array}
\right.
\end{eqnarray*}

\item When $m_p\neq0$, from Lemmas \ref{lem-32A2} and \ref{lem-32A1}, we can easily get
$$
\sum_{z \in \gf(p)}\zeta_p^{-\frac{m_p}{4}z^2}=\bar{\eta}(-m_p)\sqrt{p^*}.
$$
Thus, Equation (\ref{eqn-a02}) becomes
\begin{eqnarray*}
N_{(a,0)}
&=&
\left\{ \begin{array}{ll}
p^{m-2} +
p^{-2}\sum_{y \in \gf(p)^*} \sigma_y \left(\zeta_p^{-a}\bar{\eta}(-m_p)\sqrt{p^*}^{m+1}\right)    & \mbox{ if $m$ is odd,} \\
p^{m-2} +
p^{-2}\sum_{y \in \gf(p)^*} \sigma_y\left(\zeta_p^{-a}(-p^{\frac{m}{2}})\bar{\eta}(-m_p)\sqrt{p^*}\right)       & \mbox{ if $m\equiv2 (\rm{~mod~}4)$,} \\
p^{m-2} +
p^{-2}\sum_{y \in \gf(p)^*} \sigma_y\left(\zeta_p^{-a}(-p^{\frac{m}{2}+1})\bar{\eta}(-m_p)\sqrt{p^*}\right) & \mbox{ if $m\equiv0 (\rm{~mod~}4)$.}
\end{array}
\right.\\
&=&
\left\{ \begin{array}{ll}
p^{m-2} +
p^{-2}\bar{\eta}(-m_p)\sqrt{p^*}^{m+1}\sum_{y \in \gf(p)^*} \sigma_y \left(\zeta_p^{-a}\right)    & \mbox{ if $m$ is odd,} \\
p^{m-2}
-p^{\frac{m}{2}-2}\bar{\eta}(-m_p)\sum_{y \in \gf(p)^*} \sigma_y\left(\zeta_p^{-a}\sqrt{p^*}\right)       & \mbox{ if $m\equiv2 (\rm{~mod~}4)$,} \\
p^{m-2}
-p^{\frac{m}{2}-1}\bar{\eta}(-m_p)\sum_{y \in \gf(p)^*} \sigma_y\left(\zeta_p^{-a}\sqrt{p^*}\right) & \mbox{ if $m\equiv0 (\rm{~mod~}4)$.}
\end{array}
\right.
\end{eqnarray*}
\end{enumerate}
The desired conclusions then follow from the facts that
$$
\sum_{y \in \gf(p)^*} \sigma_y(\zeta_p^{-a})=
\left\{ \begin{array}{ll}
p-1   & \mbox{ if $a=0$ ,} \\
-1          & \mbox{ if $a\neq 0$,}
\end{array}
\right.
$$
and
$$
\sum_{y \in \gf(p)^*} \sigma_y(\zeta_p^{-a}\sqrt{p^*})=
\left\{ \begin{array}{ll}
0   & \mbox{ if $a=0$ ,} \\
{-a \overwithdelims () p}p^*          & \mbox{ if $a\neq 0$.}
\end{array}
\right.
$$
\end{proof}

\begin{lemma}\label{lem-F3}
Suppose that  $m_p\neq 0$. Let
$$
\bar{N}_0=|\{x \in \gf(q): \tr(x^{p+1})-\frac{1}{m_p}\tr(x)^2=0\}|.
$$
Then
\begin{eqnarray*}
\bar{N}_0 = \left\{ \begin{array}{ll}
p^{m-1}+p^{-1}\bar{\eta}(-m_p)\sqrt{p^*}^{m+1}(p-1)                               & \mbox{ if $m$ is odd,} \\
p^{m-1}                                                             & \mbox{ if $m$ is even.}
\end{array}
\right.
\end{eqnarray*}
\end{lemma}
\begin{proof}
By definition, we have
\begin{eqnarray}\label{eqn-NN0}
\bar{N}_0 &=& \frac{1}{p} \sum_{x \in \gf(q)} \sum_{y \in \gf(p)} \zeta_p^{y(\tr(x^{p+1})-\frac{1}{m_p}\tr(x)^2)} \nonumber \\
&=& p^{m-1}+ \frac{1}{p} \sum_{y \in \gf(p)^*} \sum_{x \in \gf(q)}   \zeta_p^{y(\tr(x^{p+1})-\frac{1}{m_p}\tr(x)^2)}  \nonumber \\
&=& p^{m-1}+ \frac{1}{p}\sum_{y \in \gf(p)^*} \sigma_y(\sum_{x \in \gf(q)}   \zeta_p^{\tr(x^{p+1})-\frac{1}{m_p}\tr(x)^2}).
\end{eqnarray}

Since the
Fourier expansion of $\zeta_p^{-\frac{1}{m_p}X^2}$ can be  expressed as
$$
\zeta_p^{-\frac{1}{m_p}X^2}=\sum_{z\in\gf(p)}a_z \zeta_p^{~zX}
$$
for any $X\in \gf(p)$, we have
\begin{eqnarray}\label{eqn-az}
a_z=\frac{1}{p}\sum_{X\in \gf(p)}\zeta_p^{-\frac{1}{m_p}X^2-zX}
\end{eqnarray}
for any $z\in \gf(p)$ and
\begin{eqnarray*}
\zeta_p^{-\frac{1}{m_p}\tr(x)^2}=\sum_{z\in\gf(p)}a_z \zeta_p^{z\tr(x)}.
\end{eqnarray*}
Therefore, Equation (\ref{eqn-NN0}) becomes
\begin{eqnarray}\label{eqn-NN01}
\bar{N}_0
&=& p^{m-1}+ \frac{1}{p}\sum_{y \in \gf(p)^*} \sigma_y(\sum_{z\in\gf(p)}  a_z \sum_{x \in \gf(q)}\zeta_p^{\tr(x^{p+1})+z\tr(x)}) \nonumber\\
&=& \left\{ \begin{array}{ll}
p^{m-1}+p^{-1}\sum_{y \in \gf(p)^*} \sigma_y( \sqrt{p^*}^m \sum_{z\in\gf(p)} a_z \zeta_p^{-\frac{m_p}{4}z^2})                               & \mbox{ if $m$ is odd ,} \\
p^{m-1}+p^{-1}\sum_{y \in \gf(p)^*} \sigma_y( -p^{\frac{m}{2}} \sum_{z\in\gf(p)} a_z \zeta_p^{-\frac{m_p}{4}z^2})        & \mbox{ if $m\equiv2 (\rm{~mod~}4)$,} \\
p^{m-1}+p^{-1}\sum_{y \in \gf(p)^*} \sigma_y( -p^{\frac{m+1}{2}} \sum_{z\in\gf(p)} a_z \zeta_p^{-\frac{m_p}{4}z^2})            & \mbox{ if $m\equiv0 (\rm{~mod~}4)$.}
\end{array}
\right.
\end{eqnarray}
However, by Equation (\ref{eqn-az}) and Lemma \ref{lem-32A2} we obtain
\begin{eqnarray*}
\sum_{z\in\gf(p)} a_z \zeta_p^{-\frac{m_p}{4}z^2}
&=& \sum_{z\in\gf(p)} \frac{1}{p}\sum_{X\in \gf(p)}\zeta_p^{-\frac{1}{m_p}X^2-zX-\frac{m_p}{4}z^2}\\
&=& \sum_{z\in\gf(p)} \frac{1}{p}\bar{\eta}(-\frac{1}{m_p})G(\bar{\eta},\bar{\chi}_1)\\
&=& \bar{\eta}(-m_p)\sqrt{p^*}.
\end{eqnarray*}
Hence, Equation (\ref{eqn-NN01}) becomes
\begin{eqnarray*}
\bar{N}_0
&=& \left\{ \begin{array}{ll}
p^{m-1}+p^{-1} \sqrt{p^*}^{m+1} \bar{\eta}(-m_p) \sum_{y \in \gf(p)^*} \sigma_y(1)                               & \mbox{ if $m$ is odd ,} \\
p^{m-1}+p^{-1} (-p^{\frac{m}{2}})\bar{\eta}(-m_p)  \sum_{y \in \gf(p)^*} \sigma_y( \sqrt{p^*})        & \mbox{ if $m\equiv2 (\rm{~mod~}4)$,} \\
p^{m-1}+p^{-1}(-p^{\frac{m+1}{2}})\bar{\eta}(-m_p) \sum_{y \in \gf(p)^*} \sigma_y(\sqrt{p^*})            & \mbox{ if $m\equiv0 (\rm{~mod~}4)$.}
\end{array}
\right.
\end{eqnarray*}
The desired conclusions then follow from the fact that
$$
\left\{ \begin{array}{ll}
\sigma_y(1) =1,\\
\sum_{y \in \gf(p)^*} \sigma_y(\sqrt{p^*})=0.
\end{array}
\right.
$$
\end{proof}

\subsection{The proof of Theorems \ref{thm-five},  \ref{thm-three} and \ref{thm-four}}
By definition, the code length of $\C_D$ is $n = |D|=n_0-1$, where $n_0$ was defined in Lemma \ref{lem-32B2}. This means that Equations (\ref{eqn-codelenth1}) and (\ref{eqn-codelenth2}) follow.

For each $b \in \gf(q)^*$, define
\begin{eqnarray}\label{eqn-mcodeword}
\bc_{b}=(\tr(bd_1), \,\tr(bd_2), \,\ldots, \,\tr(bd_n)),
\end{eqnarray}
where $d_1, d_2, \ldots, d_n$ are the elements of $D$.
Then the Hamming weight $\wt(\bc_b)$ of $\bc_b$ is
\begin{eqnarray}\label{eqn-wcb}
\wt(\bc_b)=n_0-N(b),
\end{eqnarray}
where $n_0$ and $N(b)$ were defined before. By Lemmas \ref{lem-32B2} and \ref{lem-32B4}, we have $\wt(\bc_b)=n_0-N(b)>0$ for each $b\in \gf(q)^*$. This means that the code $\C_D$ has $q$ distinct codewords. Hence, the dimension of the code $\C_D$ is $m$.

Next we shall prove the multiplicities $A_{w_i}$ of the
codewords with weight $w_i$ in $\C_D$, and will distinguish the following three cases.

\begin{enumerate}
  \item The case that $m$ is odd.

  It follows from Lemma \ref{lem-Coulter1} that $L(x)=x^{p^2}+x$ is a permutation polynomial over $\gf(q)$. Thus the unique solution $x_b$  of $L(x)=-b^p$ runs through $\gf(q)^*$ when $b$ runs through $\gf(q)^*$.
  \begin{description}
     \item[$\bullet$] When $m_p=0$, the desired conclusion of Table  \ref{tab-five0} follows from Equation (\ref{eqn-wcb}) and Lemmas \ref{lem-32B2}, \ref{lem-32B4}, \ref{lem-F1} and \ref{lem-F2}.
     \item[$\bullet$] When $m_p\neq 0$, we only give the proof for the case $\bar{\eta}(m_p)=1$ and omit the proof for the case $\bar{\eta}(m_p)=-1$ whose proof is similar. Suppose that $m_2=1$ and $\bar{\eta}(m_p)=1$, then from Lemmas \ref{lem-32B2} and \ref{lem-32B4} we obtain
\begin{eqnarray*}\label{eqn-wb}
\wt(\bc_b)
&=&n_0-N(b) \nonumber \\
&=&\left\{ \begin{array}{ll}
B_1     & \mbox{ if $\tr(x_b^{p+1})=0$ and $\tr(x_b)=0$,} \\
B_1 +Bp              & \mbox{ if $\tr(x_b^{p+1})=0$ and $\tr(x_b)\neq 0$,} \\
B_1 +B(p+1)            & \mbox{ if $\tr(x_b^{p+1}) \in \textup{SQ}$ and $\tr(x_b)=0$,}\\
B_1 +B(p-1)              & \mbox{ if $\tr(x_b^{p+1})\in \textup{N\textup{SQ}}$,}\\
B_1 +B ~or~B_1 +B(p+1)    & \mbox{ if $\tr(x_b^{p+1})\in \textup{SQ}$ and $\tr(x_b)\neq 0$,}
\end{array}
\right.
\end{eqnarray*}
where $B_1={p}^{m-2}(p-1)$ and $B=(-1)^{\frac{p-1}{2}\frac{m-1}{2}}p^{\frac{m-3}{2}}$.
Therefore, the weight $\wt(\bc_b)$ satisfies
$$
\wt(\bc_b) \in \{B_1 , B_1 +Bp, B_1 +B(p+1), B_1 +B(p-1), B_1 +B\}
$$
for each $b \in \gf(q)^*$.
Define
\begin{eqnarray*}
&& w_1=B_1, \\
&& w_2=B_1 +Bp, \\
&& w_3=B_1 +B(p+1),\\
&& w_4=B_1 +B(p-1), \\
&& w_5=B_1 +B.
\end{eqnarray*}

We now determine the number $A_{w_i}$ of the codewords with weight $w_i$ in $\C_{D}$. By Lemmas \ref{lem-F1} and \ref{lem-F2}, we can directly determine
\begin{eqnarray}\label{eqn-A124}
\left\{
\begin{array}{lll}
A_{w_1} &=& N_{(0,0)}-1=p^{m-2}-1+(p-1)B, \\
A_{w_2} &=& (N_0-1)-(N_{(0,0)}-1)=B_1-(p-1)B, \\
A_{w_4} &=& \frac{p-1}{2}N_a=\frac{p-1}{2}(p^{m-1}-Bp),
\end{array}
\right.
\end{eqnarray}
where $a$ is a nonsquare in $\gf(p)^*$. Since $0 \not\in D$, the minimum distance $d^{\perp}$ of the dual code $\C_D^{\perp}$ of $\C_D$ cannot be $1$. This means that the minimum weight of the dual code $\C_{D}^\perp$ is at least $2$. The first two Pless Power Moments \cite[p.260]{HP} lead to the following system of equations:
\begin{eqnarray}\label{eqn-wtdsemibentfcode5}
\left\{
\begin{array}{lll}
A_{w_1}+A_{w_2}+A_{w_3}+A_{w_4}+A_{w_5} &=& p^m-1, \\
w_1A_{w_1}+w_2A_{w_2}+w_3A_{w_3}+w_4A_{w_4}+w_5A_{w_5} &=& (p-1)np^{m - 1},
\end{array}
\right.
\end{eqnarray}
where the code length $n=n_0-1=p^{m-1}-1+pB$. Solving the system of equations in (\ref{eqn-wtdsemibentfcode5}) yields
\begin{eqnarray*}
\left\{
\begin{array}{lll}
A_{w_3}&=& \frac{1}{2}(p-1)(p-2)p^{\frac{m-3}{2}}(p^{\frac{m-1}{2}}-(-1)^{\frac{p-1}{2}\frac{m-1}{2}}), \\
A_{w_5}&=& (p-1)p^{m-2}+(-1)^{\frac{p-1}{2}\frac{m-1}{2}}(p-1)^2p^{\frac{m-3}{2}}.
\end{array}
\right.
\end{eqnarray*}
This completes the proof of the weight distribution of Table \ref{tab-five1}.
\end{description}

\item The case that $m \equiv 2 \pmod{4}$.

It follows from Lemma \ref{lem-Coulter2} that $L(x)=x^{p^2}+x$ is a permutation polynomial over $\gf(q)$. Similar to the analysis for the case $m$ being odd, we have the following.
\begin{description}
   \item[$\bullet$] When $m_p=0$, it is clear that the desired conclusion of Table \ref{tab-three0} follows from  Equation (\ref{eqn-wcb}) and Lemmas \ref{lem-32B2}, \ref{lem-32B4}, \ref{lem-F1} and  \ref{lem-F2}.
  \item[$\bullet$]  When $m_p \neq 0$, by Lemmas \ref{lem-32B2} and \ref{lem-32B4} we have
\begin{eqnarray}\label{eqn-wb3}
\wt(\bc_b)
&=&n_0-N(b) \nonumber \\
&=&\left\{ \begin{array}{ll}
(p-1)p^{m-2}   & \mbox{ if $\tr(x_b^{p+1})=0$ and $\tr(x_b)=0$,} \\
(p-1)p^{m-2}+p^{\frac{m}{2}-1}      & \mbox{ if $\tr(x_b^{p+1})=0$ and $\tr(x_b)\neq 0$} \\
~~ & \mbox{or $\tr(x_b^{p+1})\neq 0$ and $A= 0$,}\\
(p-1)p^{m-2} ~or ~ (p-1)p^{m-2}+2p^{\frac{m}{2}-1} & \mbox{ if $\tr(x_b^{p+1})\neq 0$ and $A\neq 0$,}
\end{array}
\right.
\end{eqnarray}
where $A=-\frac{m_p}{4}+\frac{\tr(x_b)^2}{4\tr(x_b^{p+1})}$. Therefore, the weight $\wt(\bc_b)$ satisfies
$$
\wt(\bc_b) \in \{(p-1)p^{m-2},(p-1)p^{m-2}+p^{\frac{m}{2}-1},(p-1)p^{m-2}+2p^{\frac{m}{2}-1}\}
$$
for each $b \in \gf(q)^*$.
Define
\begin{eqnarray*}
&& w_1=(p-1)p^{m-2}, \\
&& w_2=(p-1)p^{m-2}+p^{\frac{m}{2}-1}, \\
&& w_3=(p-1)p^{m-2}+2p^{\frac{m}{2}-1}.
\end{eqnarray*}

We now determine the number $A_{w_i}$ of the codewords with weight $w_i$ in $\C_{D}$. Note that
\begin{eqnarray*}
& & |\{b\in \gf(q): \tr(x_b^{p+1})\neq 0 ~and~ A= 0 \}|\\
& &=|\{b\in \gf(q): A= 0 \}|-|\{b\in \gf(q): \tr(x_b^{p+1})=0 ~and~ A\neq 0 \}| \\
& &=|\{x\in \gf(q): \tr(x^{p+1})-\frac{1}{m_p}\tr(x)^2 =0\}|-|\{x\in \gf(q): \tr(x^{p+1})=0 ~and~ \tr(x)=0 \}|,
\end{eqnarray*}
as $L(x)=x^{p^2}+x$ is a permutation polynomial over $\gf(q)$ for this case $m\equiv 2 \pmod{4}$.
Thus from Equation (\ref{eqn-wb3}) and Lemmas \ref{lem-F1}, \ref{lem-F2} and \ref{lem-F3}, we can directly obtain
\begin{eqnarray*}
A_{w_2}=(N_0-N_{(0,0)})+(\bar{N}_0-N_{(0,0)})=(p-1)(2p^{m-2}-p^{\frac{m}{2}-1}).
\end{eqnarray*}
Since the minimum weight of the dual code $\C_{D}^\perp$ is at least $2$, from the first two Pless Power Moments \cite[p.260]{HP} we can compute $A_{w_1}$ and $A_{w_3}$. This completes the proof of the weight distribution of Table \ref{tab-three1}.
\end{description}

\item The case that $m \equiv 0 \pmod{4}$.

\begin{description}
\item[$\bullet$]
If $m_p=0$, then by Lemmas \ref{lem-32B2} and \ref{lem-32B4} we have
\begin{eqnarray*}\label{eqn-wb1}
\wt(\bc_b)&=&n_0-N(b)\\
&=&\left\{ \begin{array}{ll}
B_1-(p-1)^2p^{\frac{m}{2}-1}     & \mbox{ if $b\notin \textup{Im}(L)$,} \\
B_1    & \mbox{ if $b\in \textup{Im}(L)$, $\tr(x_b^{p+1})=0$ and $\tr(x_b)=0$,} \\
B_1-(p-1)p^{\frac{m}{2}}       & \mbox{ if $b\in \textup{Im}(L)$, $\tr(x_b^{p+1})=0$ and $\tr(x_b)\neq 0$} \\
~~       & \mbox{~~~~~~~~~~~~~ or $\tr(x_b^{p+1})\neq 0$ and $\tr(x_b)= 0$,} \\
B_1-(p-2)p^{\frac{m}{2}} & \mbox{ if $b\in \textup{Im}(L)$, $\tr(x_b^{p+1})\neq 0$ and $\tr(x_b)\neq 0$,}
\end{array}
\right.
\end{eqnarray*}
where $B_1=p^{m-2}(p-1)$.

Recall that $L(x)=x^{p^2}+x$. If $b\in \textup{Im}(L)$, then this means that $L(x)=-b^p$ is solvable. We point out the following facts:
\begin{itemize}
  \item If $m\equiv 0 \pmod{4}$, then $L(x)$ is not a permutation polynomial over $\gf(q)$. However, for any $b\in \gf(q)$, if $L(x)=-b^p$ is solvable in this case $m\equiv 0 \pmod{4}$, then it has $p^2$ solutions.
  \item If $m\equiv 0 \pmod{4}$, then the number of $b$ is $p^{m-2}$ such that $L(x)=-b^p$ is solvable when $b$ runs through $\gf(q)$.
\end{itemize}
Therefore, we have
\begin{eqnarray*}
\wt(\bc_b)=\left\{ \begin{array}{ll}
p^{m-2}(p-1)-(p-1)^2p^{\frac{m}{2}-1}                               & \mbox{ occurs $p^m-p^{m-2}$ times,} \\
p^{m-2}(p-1)                                                        & \mbox{ occurs $p^{-2}N_{(0,0)}-1$ times,} \\
(p-1)p^{\frac{m}{2}}(p^{\frac{m}{2}-2}-1)                               & \mbox{ occurs $(N_0+p^{m-1}-2N_{(0,0)})p^{-2}$ times,} \\
(p-1)p^{m-2}-(p-2)p^{\frac{m}{2}}                              & \mbox{ occurs $(p^m-p^{m-1}-N_0+N_{(0,0)})p^{-2}$ times,}
\end{array}
\right.
\end{eqnarray*}
when $b$ runs through $\gf(q)^*$, where $N_0$ and $N_{(0,0)}$ were defined before. The desired conclusion of Table \ref{tab-four0} then follows from Lemmas \ref{lem-F1} and \ref{lem-F2}.

\item[$\bullet$] If $m_p\neq 0$, then from Lemmas \ref{lem-32B2} and \ref{lem-32B4} we get
\begin{eqnarray} \label{eqn-w13}
\wt(\bc_b)
&=&n_0-N(b)  \nonumber \\
&=&\left\{ \begin{array}{ll}
(p-1)p^{\frac{m}{2}-1}+B_1     & \mbox{ if $b\notin \textup{Im}(L)$,} \\
B_1    & \mbox{ if $b\in \textup{Im}(L)$, $\tr(x_b^{p+1})=0$ and $\tr(x_b)=0$,} \\
B_1+p^{\frac{m}{2}}      & \mbox{ if $b\in \textup{Im}(L)$, $\tr(x_b^{p+1})=0$ and $\tr(x_b)\neq 0$} \\
~~       & \mbox{~~~~~~~~~~~~~ or $\tr(x_b^{p+1})\neq 0$ and $A=0$,} \\
B_1 ~or ~ B_1+2p^{\frac{m}{2}}  & \mbox{ if $b\in \textup{Im}(L)$, $\tr(x_b^{p+1})\neq 0$ and $A\neq 0$,}
\end{array}
\right.
\end{eqnarray}
where $B_1=p^{m-2}(p-1)$.
Therefore, the weight $\wt(\bc_b)$ satisfies
$$
\wt(\bc_b) \in \{(p-1)p^{\frac{m}{2}-1}+B_1,~B_1,~B_1 +p^{\frac{m}{2}},~B_1 +2 p^{\frac{m}{2}}\}
$$
for each $b \in \gf(q)^*$.
Define
\begin{eqnarray*}
&& w_1=(p-1)p^{\frac{m}{2}-1}+B_1, \\
&& w_2=B_1, \\
&& w_3=B_1 +p^{\frac{m}{2}},\\
&& w_4=B_1 +2p^{\frac{m}{2}}.
\end{eqnarray*}

We now determine the number $A_{w_i}$ of the codewords with weight $w_i$ in $\C_{D}$. It is clear that
$$
A_{w_1}=p^m-p^{m-2}.
$$
By Equation (\ref{eqn-w13}) and Lemmas \ref{lem-F1}, \ref{lem-F2} and \ref{lem-F3}, we can directly obtain
\begin{eqnarray*}
\begin{array}{lll}
A_{w_3} &=& p^{-2}(N_0-N_{(0,0)}+\bar{N}_0),\\
&=&(p-1)(2p^{m-4}-p^{\frac{m}{2}-2}).
\end{array}
\end{eqnarray*}
Since the minimum weight of the dual code $\C_{D}^\perp$ is at least $2$, the first two Pless Power Moments \cite[p.260]{HP} lead to the following system of equations:
\begin{eqnarray}\label{eqn-wtdsemibentfcode6}
\left\{
\begin{array}{lll}
A_{w_1}+A_{w_2}+A_{w_3}+A_{w_4}&=& p^m-1, \\
w_1A_{w_1}+w_2A_{w_2}+w_3A_{w_3}+w_4A_{w_4} &=& (p-1)np^{m - 1},
\end{array}
\right.
\end{eqnarray}
where $n=n_0-1=p^{m-1}+p^{\frac{m}{2}}-1$. Solving the system of equations in (\ref{eqn-wtdsemibentfcode6}) gives
\begin{eqnarray*}
\left\{
\begin{array}{lll}
A_{w_2}&=& p^{m-4}+\frac{p-1}{2}(p^{\frac{m}{2}-1}+p^{m-3})-1, \\
A_{w_4}&=& \frac{1}{2}(p-1)(p-2)(p^{m-4}-p^{\frac{m}{2}-2}).
\end{array}
\right.
\end{eqnarray*}
This completes the proof of the weight distribution of Table \ref{tab-four1}.
\end{description}
\end{enumerate}

Summarizing all the conclusions above completes the proofs of Theorems \ref{thm-five}, \ref{thm-three} and \ref{thm-four}.

\section{Concluding remarks}\label{sec-concluding}

In this paper, we presented a class of linear codes $\C_D$ with a few weights and completely determined their weight distributions. The result showed that they have at most five weights. Particularly, the codes $\C_D$ presented in this paper have three weights when $m \geq 6$ and $m\equiv 2 \pmod{4}$. Many classes of linear codes with a few weights were constructed (see, for example, \cite{DingDing2,ZLFH2015,TLQZH2015,CG84,FL07,LiYueLi,CDY05,ZD14,DLLZ}). When $m$ is odd, the code $\C_D$ presented in this paper has the same parameter and weight distribution as that of $\C_D$ presented in \cite{LWL2015}. But they have different defining set $D$.

Any linear code over $\gf(p)$ can
be used to obtain secret sharing schemes \cite{CDY05,YD06}. In order to
obtain secret sharing schemes with interesting access structures, one would like to have linear codes $\C$ such that
$w_{\min}/w_{\max} > \frac{p-1}{p}$ \cite{YD06}, where $w_{\min}$ and $w_{\max}$ denote the minimum and maximum
nonzero weight of the linear code.

When $m \geq 5$ is odd and $m_p=0$, the code $\C_D$ of Section \ref{sec-main} satisfies that
\begin{eqnarray*}
\frac{w_{\min}}{w_{\max}} = \frac{p^{m-2}-p^{(m-3)/2}}{p^{m-2}+p^{(m-3)/2}} > \frac{p-1}{p}.
\end{eqnarray*}

When $m \geq 5$ is odd and $m_p \neq 0$, the code $\C_D$ of Section \ref{sec-main} satisfies that
\begin{eqnarray*}
\frac{w_{\min}}{w_{\max}} = \frac{(p-1)p^{m-2}-(p+1)p^{(m-3)/2}}{(p-1)p^{m-2}+(p+1)p^{(m-3)/2}} > \frac{p-1}{p}.
\end{eqnarray*}

When $m \geq 6$ , $m\equiv 2 \pmod{4}$ and $m_p=0$, the code $\C_D$ of Section \ref{sec-main} satisfies that
\begin{eqnarray*}
\frac{w_{\min}}{w_{\max}} = \frac{p^{m-2}-p^{m/2-1}}{p^{m-2}} > \frac{p-1}{p}.
\end{eqnarray*}

When $m \geq 6$ , $m\equiv 2 \pmod{4}$ and $m_p \neq 0$, the code $\C_D$ of Section \ref{sec-main} satisfies that
\begin{eqnarray*}
\frac{w_{\min}}{w_{\max}} = \frac{(p-1)p^{m-2}}{(p-1)p^{m-2}+2p^{m/2-1}} > \frac{p-1}{p}.
\end{eqnarray*}

When $m \geq 6$ , $m\equiv 0 \pmod{4}$ and $m_p=0$, the code $\C_D$ of Section \ref{sec-main} satisfies that
\begin{eqnarray*}
\frac{w_{\min}}{w_{\max}} = \frac{p^{m-2}-p^{m/2}}{p^{m-2}} > \frac{p-1}{p}.
\end{eqnarray*}

When $m \geq 6$ , $m\equiv 0 \pmod{4}$ and $m_p \neq 0$, the code $\C_D$ of Section \ref{sec-main} satisfies that
\begin{eqnarray*}
\frac{w_{\min}}{w_{\max}} = \frac{(p-1)p^{m-2}}{(p-1)p^{m-2}+2p^{m/2}} > \frac{p-1}{p}.
\end{eqnarray*}

Hence, the linear codes $\C_D$ presented in this paper satisfy the condition that
$w_{\min}/w_{\max} > \frac{p-1}{p}$, and can be employed to obtain secret sharing schemes
with interesting access structures \cite{ADHK,CDY05,DingDing2,YD06}.


\end{document}